\documentclass[onecolumn]{IEEEtran}
\usepackage{amsmath,paralist,amsthm,comment,amssymb}

\usepackage{cite}

\ifCLASSINFOpdf
\else
\fi
\usepackage{times,color}

\usepackage[hidelinks,colorlinks   = true,urlcolor     = blue, 
  linkcolor    = blue, 
  citecolor   = red 
  ]{hyperref}
\usepackage{graphicx}
\usepackage[justification=raggedright]{caption}
\usepackage{subcaption}

\usepackage{pgfplots}
\usepackage{epsfig}
\usepackage{float}

\newcommand{\wt}[1]{\text{wt}({#1})}
\newcommand{\mc}[1]{\mathcal{#1}}

\newcommand{\nix}[1]{}
\newcommand{\ket}[1]{|#1\rangle}

\usepackage{newfloat,algcompatible}
\usepackage[size=small]{caption}
\usepackage{etoolbox}

\AtEndEnvironment{algorithm}{\noindent\hrulefill\par\nobreak\vskip-8pt}

\DeclareFloatingEnvironment[
    fileext=loa,
    listname=List of Algorithms,
    name=Algorithm,
    placement=tbhp,
]{algorithm}
\DeclareCaptionFormat{algorithms}{\vskip-15pt\hrulefill\par#1#2#3\vskip-6pt\hrulefill}
\algblock[Input]{Input}{EndInput}
\algblockdefx[Input]{Input}{EndInput}%
    [1]{\textbf{Input} #1}%
    {}    
    
\renewcommand{\algorithmicrequire}{\textbf{Input:}}
\renewcommand{\algorithmicensure}{\textbf{Output:}}

\AtEndEnvironment{construction}{\noindent\hrulefill\par\nobreak\vskip-8pt}
\DeclareFloatingEnvironment[
    fileext=loa,
    listname=List of Algorithms,
    name=Construction,
    placement=tbhp,
]{construction}
\newtheorem{theorem}{Theorem}
\newtheorem{corollary}[theorem]{Corollary}
\newtheorem{lemma}[theorem]{Lemma}

\begin{document}
%
\title{Decoding Algorithms for Hypergraph Subsystem Codes and Generalized Subsystem Surface Codes}
%
%
%
\author{Vinuta~V.~Gayatri,
        and~Pradeep~Kiran~Sarvepalli,~\IEEEmembership{Member,~IEEE}
        
\thanks{Vinuta V.~Gayatri is currently  with Madstreet Den Technologies Pvt Ltd., Bangalore, India.}
\thanks{Pradeep Kiran Sarvepalli is with Department of Electrical Engineering, Indian Institute of Technology Madras, India.}
}
\maketitle

\begin{abstract}
Topological subsystem codes can combine the advantages of both topological codes and subsystem codes.  Suchara et al. proposed a framework based on hypergraphs for construction of such codes. They also studied the performance  of some subsystem codes.  Later Bravyi et al. proposed a subsystem surface code. Building upon these works, we propose efficient decoding algorithms for large classes of  subsystem codes on hypergraphs and surfaces.  We also propose a construction of the subsystem surface codes that includes the code proposed by Bravyi et al. Our simulations for the subsystem code on the square octagon lattice resulted in a noise threshold of 1.75\%.  This is comparable to previous result of 2\% by Bombin et al. who used a different algorithm. 
\end{abstract}


%
\IEEEpeerreviewmaketitle

\section{Introduction}
\IEEEPARstart{S}{ubsystem} codes were proposed with a view to simplify quantum error correction procedures \cite{bacon06a,kribs05,kribs05b}. 
Subsystem codes have been studied extensively since their introduction \cite{aliferis06,Aly2006,poulin05,bacon06b,ps08,Bombin2009,Suchara2010,ps2012}.
To understand how they help recall that in a 
typical error correction cycle we need to i) measure  the syndrome, 
ii)  estimate the error  from the syndrome,    and iii) apply the  estimated error as correction. 
Subsystem codes can be beneficial in the first and last steps. 

Typically,  quantum codes require many body operators for measuring the syndrome.
Subsystem codes can potentially simplify the syndrome measurement process by breaking down the many qubit measurement into a simpler set of measurements, 
where each measurement involves a fewer number of qubits. The outcomes of these measurements are combined classically to obtain the 
measurement of the original operator.
Subsystem codes  also allow for a greater degree of freedom in the correction operator. 
In addition, we can tolerate errors while encoding and syndrome measurement. 

To fully leverage the advantages mentioned above we also need to develop efficient decoders for subsystem codes. 
Therefore this paper focusses on developing efficient decoders for certain classes of subsystem codes, specifically topological subsystem codes  (TSCs),  which
are suitable for fault tolerant quantum computing. 

The first construction of topological subsystem codes is due to Bombin \cite{Bombin2009}.
The stabilizer generators of these codes are local, unlike the Bacon-Shor code. 
(Large weight stabilizers are not preferred from a fault tolerance point of view.) 
They are derived from color codes and are also called  topological subsystem color codes (TSCCs).
In these codes, the syndrome measurements can be performed   by 
measuring  two qubit   operators.
Furthermore, each syndrome can be reconstructed from  $O(1)$ such local operators. 

Suchara et al.  \cite{Suchara2010} proposed a framework based on hypergraphs within which one could construct TSCCs and many other topological subsystem codes. 
They also  proposed  two step decoding algorithms for some specific subsystem codes.
Using the framework of \cite{Suchara2010} many new classes of topological subsystem codes were proposed in \cite{ps2012}.
However, no decoding algorithms were proposed for them. 

While many topological subsystem codes can be constructed using the hypergraph framework, not all TSCs can be  constructed within this framework. 
Bravyi et al. proposed a TSC outside this framework in \cite{Bravyi2012}; it was called the subsystem surface code (SSC). 
For this code, syndrome measurements require 3-qubit measurements. 
This was also suitable to a 2D implementation like the planar surface code. 
This code was built from a square lattice and it left open the question of a more general construction of such codes. 

In this paper we are concerned with the {\em problem of efficiently decoding topological subsystem codes from hypergraphs and generalized subsystem surface  codes}. 
Our paper builds on results in \cite{Suchara2010,Bravyi2012,ps2012}. First, we generalize the two step decoding algorithms of \cite{Suchara2010} to large classes of hypergraph subsystem codes.
We then propose new constructions of surface subsystem codes and develop efficient decoders for them. 
Along the way we also prove some structural results on the hypergraph subsystem codes and subsystem surface codes.

A detailed summary of our contributions is as follows:
\begin{compactenum}[(i)]
\item  For the cubic subsytem codes derived (from color codes) we 
propose a one step decoding algorithm which we call the colored matching algorithm. 
\item  We show that the TSCCs  can be decoded in a two step process, 
where the bit flip errors are corrected first and the phase flip errors are corrected by mapping them to a color code. 
This is built on  the observation in  \cite{Suchara2010} about the subsystem code on square octagon lattice.
We decode this code using this method and show that it has a threshold of 1.75\% which is comparable to the 
threshold of 2\% obtained by \cite{Bombin2012} who used a renormalization group decoder algorithm. 
\item We then study decoding algorithms for a class of quantum codes which generalize the five-squares subsystem code proposed in \cite{Suchara2010}.
Unlike the previous class of codes the decoding of these codes leads to a mapping on two copies of a surface code. 
We identify explictily the surface codes onto which the subsystem codes are mapped.

\item We propose a new construction for subsystem surface codes. 
For all these codes, syndrome  can be extracted by means of three-qubit measurements. 
Using this construction we obtain a family of subsytem codes with lower over head than the codes proposed in \cite{Bravyi2012}.

\item We show that the proposed subsystem surface codes can be decoded by mapping to a pair of surface codes.
This result generalizes the algorithm of \cite{Bravyi2012} to the proposed subsystem surface codes. 

\end{compactenum}

The rest of the work is organized as follows: In section \ref{sec:bg}, there are basic mathematical preliminaries and the constructions of 
subsystem codes required to understand the decoding algorithms we have proposed. In section \ref{sec:cssccodes}, we present two algorithms to decode cubic subsystem color codes. 
In sections \ref{sec:TSScs}, \ref{sec:NUR3HSSCs} and \ref{sec:SSSCs}, we  present different decoding algorithms for topological subsystem color codes, 
generalized five squares subsystem codes and subsystem surface codes, respectively. In section \ref{sec:results}, the simulation results for the subsystem code on the square octagon lattice 
are given. 
We conclude in section \ref{sec:conclusion}.

\section{Background}\label{sec:bg}

\subsection{Mathematical preliminaries}

In this section, we present a brief review of some graph theoretic concepts which we use in the discussions to follow.

A graph $\Gamma$  is an ordered pair ($\mathsf{V}(\Gamma)$, $\mathsf{E}(\Gamma)$) where $\mathsf{V}(\Gamma)$ is a set of vertices in $\Gamma$ and $\mathsf{E}(\Gamma)$ is a set of edges in $\Gamma$. 
If there are $m$ edges incident on a vertex $v$, then $v$ is said to have a degree of $m$. 
 A face in a graph is a region bounded by edges.

A path is a sequence of vertices $(v_1,v_2,\dots v_k)$ where  there exists an edge $(v_i,v_{i+1})$ for all $ i \in 1,2,\dots k-1$
and all vertices are distinct except possibly $v_1$ and $v_k$. A graph is called a connected graph if there exists a path between each pair of vertices. A closed path in a graph 
is called a cycle. That is, in the path $(v_1,v_2,\dots v_k)$ if $v_1 = v_k$, then it is a cycle. A cycle $\sigma$ is also represented by an ordered set of edges $(e_1,e_2,\dots e_k)$. 
Every vertex in a cycle has an even degree with respect to the edges in the cycle.

For a graph $\Gamma$ with a set of vertices $V(\Gamma)$ and set of edges $E(\Gamma)$, a matching 
 is a subgraph of $\Gamma$ such that there is at most one edge incident on each vertex $v \in V(\Gamma)$. In other words, no two edges in the matching 
 share a common vertex. Perfect matching is a matching where there is exactly one edge incident on every vertex. A weighted graph is a graph where each edge 
 is associated with a weight which is, in general, a positive number. Cost of a matching in a weighted graph is the sum of the weights of the edges of the matching. 
 Minimum weight matching is the matching such that it has the lowest cost among all the possible matchings.

The dual of a graph $\Gamma$ embedded on a surface is a graph which has a vertex for every face in $\Gamma$ and an edge between two vertices if the corresponding faces on $\Gamma$ are adjacent. 
The dual of $\Gamma$ is denoted as $\Gamma^*$. 

The medial graph of a graph $\Gamma$ is defined as a graph which has a vertex for each edge in $\Gamma$ and  an edge between two such vertices if the corresponding edges in $\Gamma$ are incident on the same vertex in $\Gamma$.
We denote the medial graph of
$\Gamma$ by $\Gamma_m$.

A hypergraph $\mathcal{H}$ is two-tuple $(\mathsf{V(}\mathcal{H}),\mathsf{E}(\mathcal{H}))$, where  $\mathsf{V(}\mathcal{H})$
 is the vertex set and $\mathsf{E}(\mathcal{H}) $ is the collection of edges. 
 An edge $e$ is a subset of $\mathsf{V(}\mathcal{H})$. 
 An edge with two vertices is called a simple edge and the one with more than two vertices is called a hyperedge. 
If $e=(u_1,u_2, \ldots, u_k)$ we say it is a rank-$k$ edge and $e$ is said to be incident on $u_i$ for $1\leq i\leq k$. 
 In this paper we  restrict our attention to 
edges of rank two or three. 

 A hypercycle  $\sigma$ is a collection of  edges in a hypergraph 
 such that every vertex in the support of the edges has an even degree with respect to the edges. 
 Figure \ref{example_hypergraph} shows cycles in a hypergraph.

\begin{figure}[h!]
\centering
\resizebox{0.1\textwidth}{!}{%
  \includegraphics{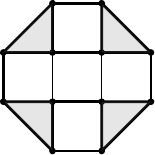}
}
\caption{A rank-$3$ hypergraph: The gray triangles are the hyperedges of the graph. 
The four inner rank-2 edges form a rank-2 cycle while the hyperedges along 
with the bold edges form a hypercycle.}
\label{example_hypergraph}
\end{figure}

We denote the vertices, edges of the graph $\Gamma$  by $\mathsf{V}(\Gamma)$, $\mathsf{E}(\Gamma)$ while the faces of the embedding of $\Gamma$ are denoted by $\mathsf{F}(\Gamma)$. 
Also $\mathsf{F}_c(\Gamma)$ and $\mathsf{E}_c(\Gamma)$ denote the set of all $c$-colored faces and $c$-colored edges respectively. If $\Gamma$ is a cubic graph such that the faces are 
3-colorable, then it is called a 2-colex. In a 2-colex, the edge colouring is induced by the faces.  An   edge connecting $c$-colored faces is also $c$-colored. In the dual graph $\Gamma^*$, 
the set of vertices corresponding to $c$-colored faces of $\Gamma$ is denoted by $\mathsf{V}_c(\Gamma^*)$.

\subsection{Stabilizer codes and subsystem codes}

\subsubsection{Stabilizer Codes}
We quickly review the stabilizer formalism, see \cite{Gottesman97,calderbank98,Gottesman} for more details. 
Denote by $\mathbb{P}_n$ the Pauli group on $n$ qubits.
 Let $S$ be a commutative subgroup of $ \mathbb{P}_n$, such that $-I\not\in S$.
The stabilizer code defined by $S$ is the joint +1-eigenspace of $S$; 
 Mathematically, it can be written as,
\[
S = \{ h \in \mathbb{P}_n \text{  }|\text{  } h\ket{\psi} = \ket{\psi} \mbox{ for all } \ket{\psi} \in Q \}
\]
$S$ is called the stabilizer of $Q$.
An $[[n,k]]$ stabilizer code encodes $k$ qubits into $n$ qubits and it
is completely characterized by $n-k$ independent stabilizer generators. 
The measurement outcome of the stabilizer generators is called the syndrome. 
 If there is a nontrivial syndrome, it indicates the presence of an error.

The centralizer of $S$, 
denoted  $C(S)$, is defined as 
\[
C(S) = \{ h \in \mathbb{P}_n \text{  }|\text{  } gh = hg \mbox{ for all } g \in S \}
\]

Since the elements of $S$ act trivially on the elements of $Q$,
errors in $\langle iI, S\rangle$ are harmless. 
If the error operator is not an element of $C(S)$, then the error can be detected. 
 An element of $C(S)$, but not an element of $\langle iI, S \rangle$, acts like a logical error and  cannot be detected.
(These relations are summarized in Fig.~\ref{st-gp-formalism}.)

The weight of a Pauli error $g$ is defined as the number of qubits on which it acts nontrivially and denoted $\wt{g}$. 
 An  $[[n,k]]$ quantum code is said to have distance $d$ where
 \begin{eqnarray}
 d=\min \{\wt{e} \mid e\in C(S)\setminus \langle iI, S\rangle  \}.
 \end{eqnarray}
 An $[[n,k,d]]$ code can detect all errors with weight up to $d-1$. 

Suppose an error $e\in \mathbb{P}_n$ occurs on the code space.  
When we measure the stabilizer generator $g\in S$, it  results in the syndrome $s \in \{0,1\}$
 where $ge =(-1)^s eg$.
 
\begin{figure*}[t!]
    \centering
    \begin{subfigure}[t]{0.5\textwidth}
        \centering
			\includegraphics{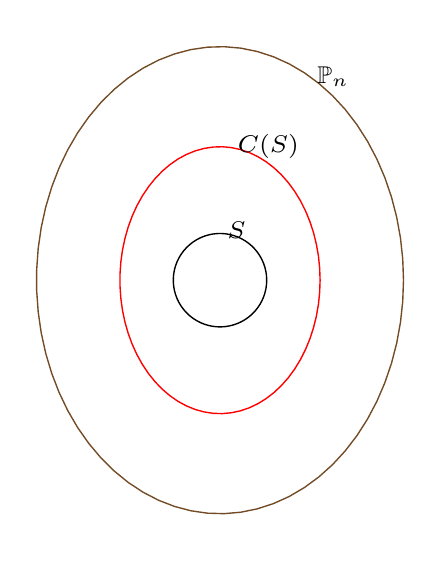}
\caption{Stabilizer codes}
    \end{subfigure}%
    ~ 
    \begin{subfigure}[t]{0.5\textwidth}
        \centering
			\includegraphics{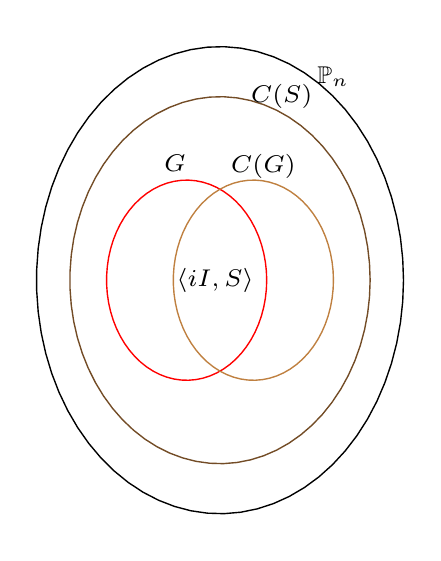}
\caption{Subsystem codes}
    \end{subfigure}
    \caption{(Color online) Relations among groups related to stabilizer codes and susbystem codes}
    \label{st-gp-formalism}
\end{figure*}

\subsubsection{Subsystem codes}
We review subsystem codes briefly, see \cite{bacon06a,Aly2006,poulin05} for an introduction.
Let $\mathcal{G}$ be a subgroup of the $n$ qubit Pauli group and $C(\mathcal{G})$ the centralizer of $\mathcal{G}$. 
Let $S$ be a (maximal) subgroup of $C(\mathcal{G})\cap \mathcal{G}$ such that $-I\not\in S$.
Then the subsystem code defined by $\mathcal{G}$ is +1-eigenspace of $S$.
The group $\mathcal{G}$ is said to be the gauge group of the subsystem code and elements of $\mathcal{G}$ gauge operators. 
Errors in $\mathcal{G}$ are considered to be harmless. 
Errors outside $C(S)$ are detectable  while errors  in $C(S)\setminus \mathcal{G}$ are  undetectable.
Suppose that $\mathcal{G} $ and  $S$ have has $2r+s$ and $s$ generators respectively. 
Then $\mathcal{G}$ defines an $[[n,k,r,d]]$ subsystem code that encodes $k=n-r-s$ qubits into $n$ qubits and has distance $d$ where
\begin{eqnarray}
d= \min \{\wt{e}\mid e\ \in C(S)\setminus \mathcal{G} \}
\end{eqnarray}
If $k=0$, then $d= \min \{\wt{e}\mid e\ \in  \mathcal{G}, e\neq \lambda I \}$.
We say the code has $r$ gauge qubits.
The subsystem code can detect all errors up to $d-1$ qubits.
(The relations between $\mathcal{G}$ and related groups are summarized in Fig.~\ref{st-gp-formalism}.)

\subsection{Topological codes}

In this paper we are interested in a class of quantum codes called topological codes \cite{kitaev03}, see also \cite{Bombin2013} for an introduction. 
Of particular relevance are the families of surface codes \cite{kitaev03} and color codes \cite{Bombin2006}. We review them briefly.

\noindent
\subsubsection{Surface codes.}
 Let $\Gamma$ be a graph embedded on a closed surface. %
Qubits are attached to the edges of $\Gamma$.
We define two types of Pauli operators as follows:
\begin{eqnarray}
A_v = \prod_{e \in \delta v} X_e \mbox{ and }
B_f = \prod_{e \in \partial f} Z_e, \label{eq:ktc-stab}
\end{eqnarray}
where $\delta v $ is the set of edges that are incident on the vertex $v$ and $\partial f $ is the set of edges that form the boundary of the face $f$. Operators of the type $A_v$ are called vertex operators while $B_f$
are called plaquette operators. 
We define the surface code on $\Gamma$ to be the stabilizer code whose stabilizer is given by 
\[
S(\Gamma) = \langle A_v,B_f | v \in V(\Gamma), f \in F(\Gamma) \rangle
\]
If $\Gamma$ has $n_v$ vertices and $n_f$  faces, then there are $n_v - 1$ and $n_f - 1$ independent 
vertex and plaquette operators in $S(\Gamma)$.
$A_v$ and $B_f$ generators. 
The surface code encodes $2g$ qubits into $n_e$ qubits, where $n_e$ is the number of edges and 
$g$ is the genus of the surface on which $\Gamma$ is embedded. 
So it is an $[[n_e,2g]]$ quantum code. 
For instance, if $\Gamma$ is embedded on a torus as in Fig.~\ref{surface}, it encodes two qubits corresponding to the two independent cycles of nontrivial homology. 

\begin{figure*}[t!]
    \centering
    \begin{subfigure}[t]{0.5\textwidth}
        \centering
			\includegraphics{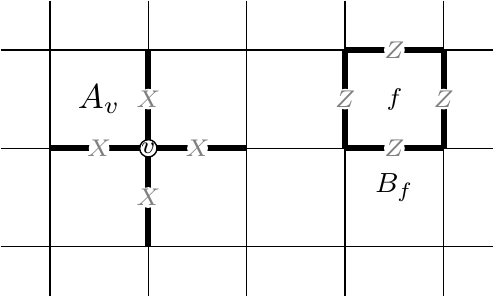}
\caption{Surface code on a square lattice. 
Qubits are on edges. Stabilizer generators $A_v$ and $B_f$
are defined as in Eq.~\eqref{eq:ktc-stab}.} \label{surface}
    \end{subfigure}%
    ~ 
    \begin{subfigure}[t]{0.5\textwidth}
        \centering
			\includegraphics{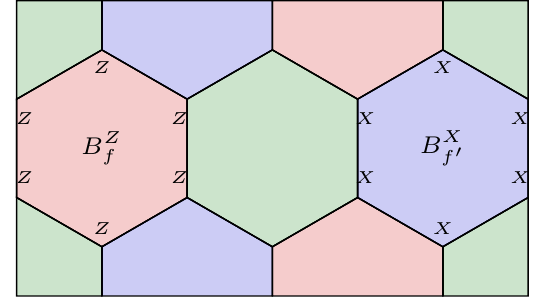}
\caption{(Color online) Color code on a hexagonal lattice. Qubits are on vertices. Stabilizer generators $B_f^X$ and $B_{f'}^Z$
are defined as in Eq.~\eqref{eq:stab_color}.} \label{honeycomb}
    \end{subfigure}
    \caption{Topological codes}
\end{figure*}

\medskip
\noindent
\subsubsection{Color codes.} A color  codes is defined  by a trivalent graph $\Gamma_2$ which is 3-face-colorable.
By convention, these colors are taken to be red, green and blue and labeled $r$, $g$ and $b$ respectively. Such a graph is also
3-edge-colorable.
The edge connecting two faces of the same color, say $r$, is also colored with $r$ and so on. In case of color codes, a qubit is attached to every vertex of the graph. There are two stabilizers associated with  each face of this lattice as given below : 
\begin{eqnarray}
B_f^X = \prod_{v \in f} X_v \mbox{  and }
B_f^Z = \prod_{v \in f} Z_v \label{eq:stab_color}
\end{eqnarray}
The stabilizer for the color code is generated by $B_f^X$, $B_f^Z$ for all $f\in\mathsf{F}(\Gamma_2)$.
The honeycomb lattice is an example of a color code, see Fig.~\ref{honeycomb}.
A method to obtain 2-colexes is given in Construction~\ref{alg:tcc-construction}.

\begin{construction}[h]
\caption{{\ensuremath{\mbox{Color codes from graphs \cite{Bombin07}}}}}\label{alg:tcc-construction}
\begin{algorithmic}
\STATE Given an arbitrary graph, color each face with $x \in \{r,g,b\}$ and split each edge into two edges. Color the faces formed after splitting the edges with $y \in \{r,g,b\} \backslash x$. Then change each vertex of degree $d$ to a face with $d$ edges and color them with $z \in \{r,g,b\} \backslash \{x,y\}$.
\end{algorithmic}
\end{construction}

Note that the 2-colex from Construction~\ref{alg:tcc-construction} has three kinds of faces: those obtained from faces, edges and vertices. They are called $f$-faces, $e$-faces and $v$-faces respectively. All the faces of a certain type are colored the same. 

\subsection{Subsystem codes from graphs and hypergraphs}
\label{ssec:hg-subsys-codes}

This subsection reviews topological subsystem codes based on cubic lattices and hypergraphs. Suchara et al. proposed a framework within which we can construct subsystem codes from cubic graphs and rank-3 hypergraphs \cite{Suchara2010}. A subsystem code is constructed from a trivalent  rank-3 hypergraph  $\mathcal{H}$ as follows. Qubits are placed on the vertices of $\mathcal{H}$.
We associate to each edge of the hypergraph a Pauli operator $K_e$ called edge operator
\begin{eqnarray}
K_e &=& \left\{ \begin{array}{lcl}g_ug_v &\text{ if } & (u,v) \in \mathsf{E}_2(\mathcal{H})\\
Z_uZ_vZ_w&\text{ if }&(u,v,w) \in \mathsf{E}_3(\mathcal{H})
\end{array}
\right.
\end{eqnarray}
 subject to the following constraints: 
\begin{eqnarray}
|e\cap e' |\leq 1 &\text{ and }& K_e K_{e'} = (-1)^{|e\cap e'|} K_{e'} K_{e}
\end{eqnarray}

From the hypergraph we derive a standard graph $\mathcal{\overline{H}}$ without hyperedges by promoting each hyperedge $(u,v,w) $ to three regular 
edges $(u,v)$, $(v,w)$ and $(w,v)$. This derived graph is used to define the gauge group $\mathcal{G}$ of a subsystem code. We define the link operators $\overline{K}_e$ as 
\begin{eqnarray}
\overline{K}_{u,v} &=& \left\{ \begin{array}{lcl}K_e &\text{ if } & (u,v) \in \mathsf{E}_2(\mathcal{H}) \\
Z_uZ_v&\text{ if }& (u,v) \subset (u,v,w) \in \mathsf{E}_3(\mathcal{H})
\end{array}\right.\\
\mathcal{G} &=& \langle \overline{K}_{u,v}\mid (u,v) \in \mathsf{E}_2(\mathcal{\overline{H}}) \rangle,
\end{eqnarray}
where $\mathcal{G}$ is the group generated by the link operators. 

The hypergraph also defines the centralizer of the gauge group, denoted by $C(\mathcal{G})$.
Let $\sigma$ be a hypercycle of $\mathcal{H}$. Then we define as
the cycle operator $W_\sigma $
\begin{eqnarray}
W_\sigma &=& \prod_{e\in \sigma } K_e\\
\text{Then }C(\mathcal{G}) &=& \langle W_\sigma\mid \sigma \text{ is a hypercycle of } \mathcal{H}\rangle
\end{eqnarray}
Using this framework we can construct many families of subsystem codes. We restrict our attention to the following 
 families of codes within this framework.

\subsubsection{Cubic subsystem color codes } \label{sssec:cubic-subsys}
Consider a cubic graph (that is also bipartite). We consider only the case when the graph $\Gamma$ is a 2-colex. 
Then $\Gamma$ is 3-face-colorable and  3-edge-colorable. We can assign the edge operator 
$K_e$ based on the color. Without loss of generality, we assume 
\begin{eqnarray}
K_e= \left\{\begin{array}{cl }X_uX_v & (u,v) \in \mathsf{E}_x(\Gamma)\\ 
Y_uY_v & (u,v) \in \mathsf{E}_y(\Gamma)\\ 
Z_uZ_v & (u,v) \in \mathsf{E}_z(\Gamma),
\end{array}\label{eq:cubic-link-ops}
\right.
\end{eqnarray}
where $x,y,z$ are three distinct colours. 
We have the following dependency among the edge operators of a cubic subsystem code. 
\begin{eqnarray}
\prod_e K_e = I\label{eq:cubic-gauge-dep}
\end{eqnarray}
As a consequence of the above assignment, a $c$-colored face is bounded by an alternating sequence of edges of color $\{x,y,z \}\setminus \{c\}$ where $c\in \{x,y,z \}$. Every cycle leads to a stabilizer. Therefore, the stabilizer associated to each  face is as follows. 
\begin{eqnarray}
B_f= \left\{\begin{array}{cl }\prod_{v\in f} X_v & f \in \mathsf{F}_x(\Gamma)\\ 
\prod_{v\in f} Y_v & f \in \mathsf{F}_y(\Gamma)\\ 
\prod_{v\in f} Z_v & f \in \mathsf{F}_z(\Gamma)
\end{array}\label{eq:cubic-stab-ops}
\right.
\end{eqnarray}
There are also additional stabilizer generators coming from homologically nontrivial cycles of $\Gamma$.  We call these codes  {\em cubic subsystem color codes} indicating that the underlying cubic graph is a 2-colex. We can easily show 
the following result about these codes. 
\begin{lemma}[Cubic subsystem color codes]\label{lm:cubic-params}
A 2-colex with $n$ vertices  gives rises to  an $[[n,0,n/2-1,2]]$ subsystem code. 
\end{lemma}
\begin{proof}
In a 2-colex the number of edges is $e=3n/2$. Therefore the number of linearly independent gauge generators is $3n/2-1$.
The number of independent stabilizer generators is given by $s=f+2g-1$. Since $n-3n/2+f=2-2g$, we have $f=n/2+2-2g$
and $s= n/2+1$. Let $r$ be the number of gauge qubits. Then $3n/2-1=2r+s$, giving $r=n/2-1$ gauge qubits. 
The dimension of the subsystem code is  $k=n-s-r=0$. Fnally the distance is two since, $\min \wt{\mathcal{G}}=2$.
\end{proof}

\subsubsection{Topological subsystem color codes }

A method to construct topological subsystem codes \cite{Bombin2009} is to start with a 2-colex and perform the following sequence of operations. 
Expand every vertex into a triangle and split every edge into a pair of edges as shown in Fig.~\ref{fig:vertex-expansion}. We identify every triangle with a rank-3 edge. The resulting hypergraph is trivalent and 3-edge-colorable. We color the rank-3 edges one color say blue, 
and the rank-2 edges incident on every hyperedge red and green in an alternating fashion as shown in Fig.~\ref{fig:vertex-expansion}. 
An  example is shown in the Fig.~\ref{tsc2color}.
Then we can assign the link operators as in Eq.~\eqref{eq:cubic-link-ops}.

\begin{figure}[htb]
\centering
  \includegraphics{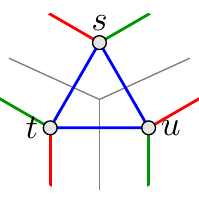}
\caption{(Color online) A vertex of the 2-colex along with its incident edges is shown (grey) and the subsequent vertex expansion shown in color. The rank-3 edge is always coloured blue. The rank-2 edges are coloured with red (solid) and green   in an alternating fashion going clockwise around a rank-3 edge.} \label{fig:vertex-expansion}
\end{figure}

\begin{figure}[h]
\centering
\includegraphics[width=80mm]{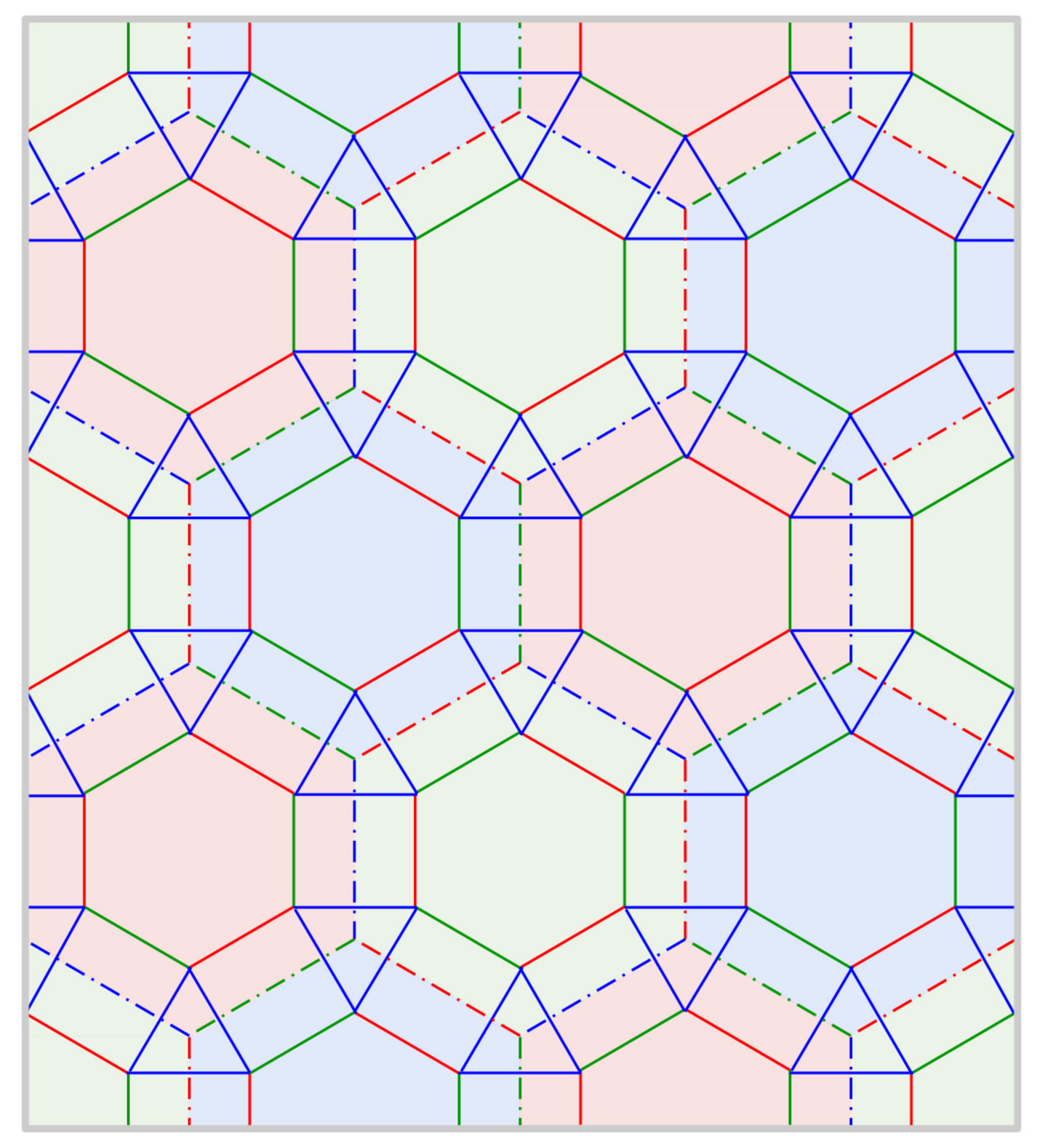}
\captionof{figure}{(Color online) Illustrating vertex expansion on the hexagonal lattice.}
\label{tsc2color}
\end{figure}

In the hypergraph generated from the 2-colex by vertex expansion, we can associate two independent cycles for each face:
i) one rank-2 cycle and ii) one hypercycle. Therefore we obtain two independent stabilizer generators for each face. The  cycles lead to the following stabilizer generators. 
\begin{eqnarray}
W_1^f   &=& \prod_{v\in f} Z_v \label{eq:rank2-stab}\\
W_2^f   &=& \prod_{(u_i,v_i,w_i) \in f} X_{u_i} Y_{v_i} Y_{w_i} \label{eq:rank3-stab}
\end{eqnarray}

The codes obtained from this construction have the following parameters. 
\begin{lemma}[Topological subsystem color codes \cite{Bombin2009}]\label{lm:tscc-params}
A 2-colex $\Gamma$, embedded on a surface of genus $g$, on vertex expansion leads to a $[[3n,2g,2n+2g-2,d\geq \ell]]$ subsystem code, where
$n$ is the number of vertices in the 2-colex and $\ell$ is the length of the smallest cycle of nontrivial homology in $\Gamma$. 
\end{lemma}

\subsubsection{Hypergraph subsystem codes}
Another class of subsystem codes  were proposed in \cite{ps2012} using the framework of Bravyi et al. 
They are also based on color codes.  The central idea behind the constructions of
\cite{ps2012} is to come up with a trivalent,  3-edge colorable hypergraph by promoting some of the edges of a 2-colex to rank-3 edges. On promoting some rank-2 edge to rank-3 edges the resulting graph is not necessarily trivalent, so additional rank-2 edges are added to make for those vertices that are deficient in degree. Since the 2-colex is 3-edge colorable, with a little care we can ensure that the resulting hypergraph is also 3-edge-colorable.

\begin{construction}[h]
\caption{\ensuremath{\mbox{Hypergraph subsystem codes from color codes \cite{ps2012}}}}\label{alg:hg-tsc-construction}
\begin{algorithmic}
\STATE
Given a 2-colex $\Gamma_2$ and a subset of $r$-colored faces $\mathsf{F} \subseteq \mathsf{F}_r(\Gamma_2)$ such that $|f| \equiv 0 \text{ mod } 4$ and $|f|>4$ for all $f \in \mathsf{F}$. Inside each $f \in \mathsf{F}$, add a face $f'$ such that $f'$ has $|f|/2$ edges. Take an alternating set of edges in the boundary of $f$ and promote them to hyperedges such that (i) the third vertex of a hyperedge comes from the face $f'$ and (ii) the hyperedges do not cross each other.
\end{algorithmic}
\end{construction}

We consider two families of hypergraph subsystem codes: 
\begin{compactenum}[i)]
\item $\mathsf{F} = \mathsf{F}_r(\Gamma_2)$ 
\item $\mathsf{F} \subsetneq \mathsf{F}_r(\Gamma_2)$
\end{compactenum}

We call the codes from first choice as uniform rank-3 hypergraph subsystem codes and the codes from the latter as generalized five-squares subsystem codes. 
These codes are explored in Sections.~\ref{sec:TSScs}~and~\ref{sec:NUR3HSSCs} respectively.

\section{Decoding Cubic subsystem color codes}\label{sec:cssccodes}

The prototype cubic subsystem code is Kitaev's honeycomb model \cite{Kitaev2006}. 
This was not originally viewed as a subsystem code. Suchara et al. \cite{Suchara2010}
made this connection to subsystem codes clear. Although the cubic subsystem codes are $[[n,0]]$ quantum codes,  
it is instructive to study them. 
For instance, Kells et al.~\cite{Kells2008} showed that that the honeycomb model can be approximated by a toric code.
Recently, Lee et al. \cite{lee2017} showed that the honeycomb model can be viewed as an approximate error correcting code. 
Suchara et al. studied the hexagonal subsystem code to obtain valuable insights into syndrome measurement using two qubit operators.
These developments indicate that it is worthwhile to study cubic subsystem codes.

We restrict our attention to the cubic subsystem color codes i.e.  cubic subsystem codes based on 2-colexes, see 
Section~\ref{sec:bg} for a brief review.
We  study the decoding of these codes and propose a number of algorithms which have natural extensions to hypergraph subsystem codes which encode a nontrivial number of logical qubits. The key ideas are much more 
transparent in the context of cubic subsystem codes.

\subsection{A  decoder for cubic subsystem color codes} 

We propose  a  decoding algorithm for the cubic subsystem codes. 
This generalises the method proposed in \cite{Suchara2010} to cubic subsystem codes. 
In this method all the $X$ errors are corrected in the first step. This could introduce some new $Z$ type errors. In the next step all the 
$Z$ errors are mapped to 
a surface code. We then decode the errors on the surface code. 
The resulting error estimate (on the surface code) is then lifted to the subsystem code leading to an error estimate for the subsystem code. This completes the error correction cycle. 

Two errors that differ by an element of the gauge group cannot be distinguished by the subsystem code. 
We consider such errors equivalent. The next lemma shows that for every error there is an equivalent error that has a  structure which simplifies the error correction. 

\begin{lemma}[Error equivalence modulo gauge group]\label{lm:-cubic-reduction}
An error on a cubic subsystem code is equivalent to another error which has
\begin{enumerate}[(a)]
\item  at most one  $X$ or $Y$ error per $z$-face 
\item  $Z$ errors on the remaining qubits
\end{enumerate} 
\end{lemma}
\begin{proof}
Let $\Gamma_2$ be the 2-colex on which the cubic subsystem code is defined. 
The $z$-colored faces of $\Gamma_2$  are disjoint and cover all the vertices of $\Gamma_2$, it suffices to show that the error on any such face is equivalent  to  a $Z$-only error. Let $f \in \mathsf{F}_z(\Gamma_2)$. 
Since all faces in 2-colex have even number of edges we can assume that 
the link operators for the edges in the  boundary of this face are given by $K_{e_1}$, \ldots, $K_{e_{2\ell}}$ where 
$e_{2i} \in \mathsf{E}_x (\Gamma_2)$ and $e_{2i+1} \in \mathsf{E}_y(\Gamma_2)$.
The product of the link operators of the first $i<2\ell$ edges will give us 
\begin{eqnarray}
K_{e_1} K_{e_2}\cdots K_{e_i} = \left\{\begin{array}{cl}X_1Z_2 \ldots Z_iX_{i+1} & i \text { odd }\\
X_1Z_2 \ldots Z_iY_{i+1} & i \text { even }\end{array} \right.
\end{eqnarray}
It follows from these equations that  $X_{i+1}$ is equivalent to $X_1Z_2\cdots Z_i$ if $i$ is odd, else 
$X_{i+1}$ is equivalent to $X_1Z_2\cdots Z_iZ_{i+1}$.
Therefore, an arbitrary error can be reduced to at most one $X$ error and a combination of $Z$ errors on the remaining qubits. 
\end{proof}

Figure \ref{ex:-cubic-reduction} shows an illustration of lemma \ref{lm:-cubic-reduction}. The blue-colored edges of the hexagon are associated with the link operator $YY$ and green-colored edges are associated with 
the link operators $XX$.The hexagon represents a $red$-colored face of a 2-colex on which the subsystem code is defined. An equivalent of an error is obtained by multiplying the error with the link operator associated 
with an adjoining edge.

\begin{figure}
\centering
  \includegraphics{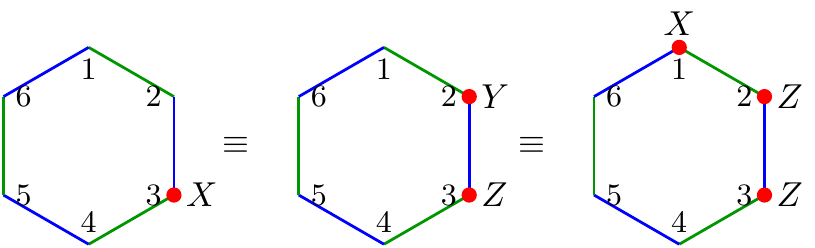}
\caption{(Color online) The green edges represent $XX$ operators and the blue edges represent $YY$ operators. Thus $X$ error on qubit 3 of the hexagonal face, ie. $X_3$ is equivalent to $Y_2Z_3$, 
which is in turn equivalent to $X_1Z_2Z_3$.}
\label{ex:-cubic-reduction}
\end{figure}

\begin{corollary}[Correcting $X$ errors up to gauge]
An $X$ error on a cubic subsystem code can be corrected up to a gauge operator by 
measuring the $Z$-stabilizers. 
\end{corollary}\label{co:cubic-x-clean}
\begin{proof}
By Lemma~\ref{lm:-cubic-reduction}, we know that every error is equivalent to an error that contains at most one $X$ or $Y$ error on every $z$-face of $\Gamma_2$. 
By Eq.~\eqref{eq:cubic-stab-ops}, the stabilizer associated to $f$ is given by 
\begin{eqnarray}
B_f & = & \prod_{e\in \partial f} K_e = \prod_{v\in f} Z_v
\end{eqnarray}
The stabilizers of distinct $z$-faces have disjoint support, therefore we can correct each of the errors independently. 
If the stabilizer of an $r$-face gives a nontrivial syndrome then we can apply an $X$ error on one of the qubits otherwise we make no correction. This process completely estimates all the $X$ errors.
\end{proof}

After correcting the bit flip errors on the cubic subsystem code we are left with only phase  errors. Strictly speaking, the residual error is equivalent to a  $Z$-only error. The $Z$ errors affect only the syndromes measured on the $x$ and $y$ faces. (The stabilizers corresponding to these faces are  $\prod_{v\in f} X_v$  and $\prod_{v\in f} Y_v$ respectively.) 
We shall show the remaining  $Z$ errors can be estimated by projecting them onto a surface code.

The  stabilizers of trivial homology    correspond to vertices in $\Gamma^\ast$. 
These vertices can be partitioned as $\mathsf{V}_x(\Gamma^\ast)$, 
 $\mathsf{V}_y(\Gamma^\ast)$, and $\mathsf{V}_z(\Gamma^\ast)$. Once the bit flip errors are corrected the nonzero syndromes can be found only on the vertices in $\mathsf{V}_x(\Gamma^\ast)\cup \mathsf{V}_y(\Gamma^\ast)$. 
 So we might as well delete the vertices in 
$\mathsf{V}_z(\Gamma^\ast)$ and consider the resulting complex. 
 Denote by  $\Gamma^{\ast\setminus z}$ the complex obtained  by (i) taking the dual of $\Gamma$ and (ii) removing all the vertices in $\Gamma^\ast$ corresponding to $z$-colored faces of $\Gamma$ and the edges incident on these vertices. (The same complex can also be obtained by contracting all the $x$ and $y$ edges of $\Gamma$ and then taking the dual.)

If $e:=(u,v)$ is a $z$-edge, then it connects two $z$-colored faces and there is a gauge group operator acting on the qubits sitting on $u$ and $v$ given by $K_e = Z_uZ_v$. Hence, a $Z$ error on the qubit $u$ is equivalent to a $Z$ error on $v$. Therefore, it is sufficient to consider either $u$ or $v$ for correcting the $Z$ errors on $u$, $v$. 
Further, every qubit is incident on a unique $z$-edge. 
So every $Z$ error can be identified with a unique $z$-edge in $\Gamma$.  
A $Z$ error on a single qubit in  $\Gamma$, causes nonzero syndrome on  $x$ and $y$ faces containing it. 
Equivalently, it will cause nonzero syndrome on two adjacent $x$ and $y$ vertices. These vertices share a $z$-edge.  
Therefore a single $Z$ error corresponds to a $z$ edge in $\Gamma^\ast$. In fact this is the same unique $z$-edge associated with the $Z$ error.

Therefore, a collection of $Z$ errors  corresponds to a collection of edges in $\Gamma^{\ast\setminus z}$. This same collection of edges can be viewed as a set of paths and cycles in $\Gamma^{\ast\setminus z}$. The end points of these paths correspond to the nonzero syndromes. Decoding reduces to finding a set of paths which terminate on the vertices with nonzero syndrome. This is precisely, the problem of decoding $Z$ errors on the toric code defined by $\Gamma^{\ast 
\setminus z}$. Thus we have shown the following result. 

\begin{theorem}[Projecting $Z$ errors onto a surface code]\label{th:Z-error-cubic-codes}
A $Z$-type error, or an equivalent error,  on a cubic subsystem code can be projected to $Z$-type error on a toric code on $\Gamma^{\ast\setminus z}$.
\end{theorem}

Once we estimate the error on the toric code on $\Gamma^{\ast\setminus z}$, we can lift this to an error on the subsystem code easily. An error on edge $(u,v)$ 
 corresponds to $Z_u$
or the equivalent error $Z_v$. 

After decoding on the surface code, we still might have some $Z$ errors left over since the toric code does not detect errors which form cycles on $\Gamma^{\ast\setminus z}$. The cubic subsystem color code can correct even such errors (modulo the gauge group) since it has $2g$ additional stabilizer generators (of nontrivial homology). 
Before we can correct such errors we must compute the syndrome on the non trivial stabilizer generators after accounting for the bit flip errors and the phase flip errors as estimate from the toric code. Then the residual error has no syndrome on the trivial stabilizers and can be  accounted by the nontrivial stabilizers alone. For each such stabilizer $S_i$ we find a $Z$ error $\hat{E}_i$ such that $\hat{E}_i$ has an even number of errors on every $x$ and $y$ face and anticommutes with just one nontrivial stabilizer. Given the syndrome $(s_1, \ldots s_{2g})$ we estimate the error as $\prod \hat{E}_i^{s_i}$. This completes the error correction on the cubic subsystem color code. We summarize the error correction procedure in the Algorithm \ref{alg:ktc-projection2}.

\begin{algorithm}[ht]
\caption{{\ensuremath{\mbox{Decoding cubic subsystem color codes by projections}}}}\label{alg:ktc-projection2}
\begin{algorithmic}[1]
\REQUIRE {A cubic subsystem code on a 2-colex $\Gamma$, syndromes on faces  $s_f,  f \in \mathsf{F}(\Gamma)$ and 
$s_{i}$ from stabilizers of nontrivial homology, $1\leq i\leq 2g$.}
\ENSURE {Error estimate $\hat{E}$ such that $\hat{E}$ has the input syndrome.}
\STATE $\hat{E}=I$
\FOR {$f\in \mathsf{F}_z(\Gamma)$}
\IF {$s_f\neq 0$}

\STATE $\hat{E}= \hat{E} X_v$ for some  $v\in f$ \COMMENT  Apply correction on exactly one vertex of $f$
\STATE Flip the syndrome of the $y$ face that  contains $v$.
\ENDIF
\ENDFOR
\STATE Project syndromes onto $\Gamma^{\ast\setminus z}$
\STATE Decode the errors on $\Gamma^{\ast\setminus z}$ using any 2D toric code decoder. Denote by $E'$, the estimate. 
\STATE Let $\Omega$ be the collection edges in $E'$.
\FOR { each edge $(u,v) \in \Omega$} \COMMENT Lift the error to subsystem code
\STATE $\hat{E} = \hat{E} Z_u$ 
\ENDFOR
\FOR {$1\leq i \leq 2g$}
\STATE Let $s_i'$ be the syndrome of $\hat{E}$ with respect to the  $i$th nonlocal stabilizer 
\STATE Update the syndrome $s_i =s_i\oplus s_i'$
\ENDFOR
\STATE Let $\hat{E}_i$ be a $Z$ error that anticommutes with all stabilizers except $i$th nontrivial stabilizer.
\STATE Return $\hat{E}=\hat{E} \prod_i \hat{E}_i^{s_i}$
\end{algorithmic}
\end{algorithm}

We now consider an example that illustrates the  complete decoding process.
Consider the subsystem code on honeycomb lattice on a torus with six hexagons, see Fig. \ref{Kitaev_honeycomb_c1}. 
This subsystem code  has seven independent stabilizers : five of trivial homology and two of nontrivial homology. 
\begin{align*}
&S_1 = Z_1 Z_5  Z_8 Z_{10} Z_7 Z_4 &S_2 = Z_2 Z_6 Z_3Z_{12} Z_9 Z_{11}\\
&S_3 = X_2 X_6  X_9 X_{11} X_8 X_5 &S_4 = X_1 X_4 X_3 X_{12} X_7 X_{10}\\
&S_5 = Y_3 Y_4  Y_7 X_{12} Y_9 Y_6 &\\
&S_6 = Y_2X_5X_8Y_{11} &S_7 = X_7 Z_{10} Y_8 X_{11} Z_9 Y_{12}
\end{align*} 
The cycles associated to  stabilizer generators of nontrivial homology are shown in the 
Fig.~\ref{Kitaev_honeycomb_c1}~and~\ref{Kitaev_honeycomb_c2}. 

\begin{center}
\begin{figure}[h!]
\centering
  \includegraphics{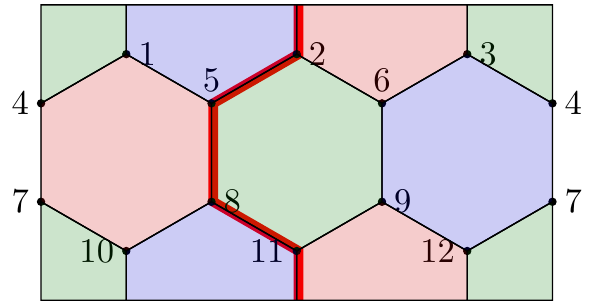}
\caption{(Color online) A nontrivial cycle $\sigma_1$ for a hexagonal subsystem code with 12 qubits}
\label{Kitaev_honeycomb_c1}
\end{figure}
\end{center}

\begin{center}
\begin{figure}[h!]
\centering
  \includegraphics{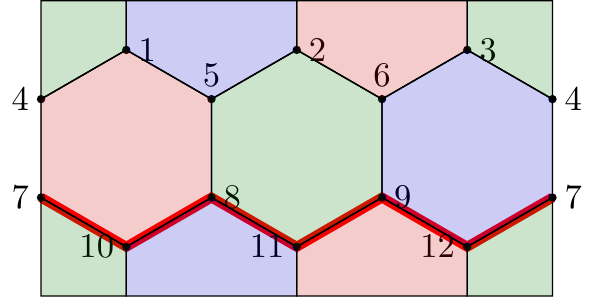}
\caption{(Color online) A nontrivial cycle $\sigma_2$ for a hexagonal subsystem code with 12 qubits}
\label{Kitaev_honeycomb_c2}
\end{figure}
\end{center}
 
 Suppose there is a $Z$-error on qubit 4 and an $X$-error on the qubit labeled 8, ie., $E=Z_4X_8$. 
 Let $s_{i}$ denote the syndrome for stabilizer $S_i$. Let the cycle associated to 
 $S_i$ be $\sigma_i$. 
 We measure the following syndromes caused by this error: 
 \begin{eqnarray*}
&&s_1 = 1 \quad s_2 = 0\\
&&s_3 = 0 \quad s_4 = 1\\
&&s_5 = 1 \\
&&s_6 = 0 \quad s_7 = 1
\end{eqnarray*} 
 
 We know that cycles $\sigma_1$ and $\sigma_2$ form $z$-faces. So we consider $s_1$ and $s_2$ to clean the $X$-errors. Since $s_1=1$, let $\hat{E}=X_1$. Remember that we can choose this $X$ error to be on any one of the six vertices forming the cycle $\sigma_1$. This makes $s_1$ equal to zero. There is no syndrome generated  on the other $z$-face, hence we do not estimate any $X$-error on $\sigma_2$. The updated syndromes are given below.
 \begin{eqnarray*}
&&s_1 = 0 \quad s_2 = 0\\
&&s_3 = 0 \quad s_4 = 1\\
&&s_5 = 1 \\
&&s_6 = 0 \quad s_7 = 1
\end{eqnarray*} 
Now, we project the syndromes
onto a surface code which is as shown in Figure \ref{fig:-Kitaev_honeycomb_2_surface} with triangular and hexagonal faces marked with black and white diagonal lines respectively. On the surface code, this set of syndromes is equivalent to the non-trivial syndromes on two triangular faces which share a common vertex 7, and zero syndrome on all the other faces. Since the traingular faces detect $Z$-type syndromes, one possible error estimate on the surface code will be $Z_7$. This error is equivalent to the error $Z_7$ or $Z_{12}$ on the cubic subsystem code. Let us consider $Z_7$. Hence the error estimate becomes $\hat{E}=X_1Z_7$. Now the syndromes are changed to 
 \begin{eqnarray*}
&&s_1 = 0 \quad s_2 = 0\\
&&s_3 = 0 \quad s_4 = 0\\
&&s_5 = 0 \\
&&s_6 = 0 \quad s_7 = 1
\end{eqnarray*} 

For correcting the error on the cycle $\sigma_7$ of nontrivial homology, we consider a $Z$-error on either qubit 7 or qubit 12. Let it be $Z_{12}$. To make the number of $Z$ errors even, on the $x$- and $y$-faces on which qubit 12 resides, let us also consider $Z_7$. So the final error estimate is $\hat{E}=X_1Z_{12}$. By inspection, we see that the errors $E$ and $\hat{E}$ are equivalent. that is, $\hat{E}$ is $E$ modulo gauge group.
 
\begin{figure}[h!]
\centering
  \includegraphics{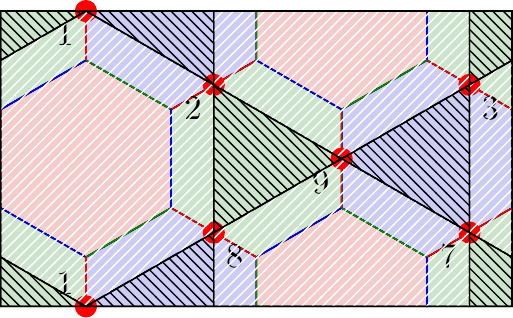}
\caption{(Color online) Mapping hexagonal subsystem code to surface code}
\label{fig:-Kitaev_honeycomb_2_surface}
\end{figure}

\subsection{A colored matching decoder for cubic subsystem color codes}

Instead of correcting the $X$ errors first and then correcting the $Z$ errors by projecting onto a surface code, we could attempt to 
decode them simultaneously. 
To present this method it is helpful to view the cubic subsystem code over the dual of 2-colex. We shall first show how to represent errors in the dual complex. 

\begin{lemma}\label{lm:error2edge}
Assuming the assignment of gauge operators as in Eq.~\eqref{eq:cubic-link-ops}, an $X$ error  corresponds to $x$-edge,
$Y$ error a $y$-edge and a $Z$ error $z$-edge  in $\Gamma^\ast$. 
\end{lemma}
\begin{proof}
As  $\Gamma$ is  trivalent,  every face is a triangle in the dual complex. The 3-face-colorability of $\Gamma$ translates to $3$-vertex colorability of $\Gamma^\ast$. While  $\Gamma^\ast$ is not 3-edge-colorable, we retain the coloring of edges based on $\Gamma$. Then every triangle in $\Gamma^\ast$ is bounded by edges of three distinct colors. Further, an $x$-edge connects $y$ and $z$ vertices of each triangle. Similarly, a $y$-edge connects $x$, $z$ vertices and $z$-edge connects $x$, $y$ vertices.

A single qubit error can flip only two checks,  therefore in $\Gamma^\ast$, a single qubit error leads to nonzero syndrome on just two vertices of a triangle. 
An $X$-error causes a nonzero syndrome on $y$ and $z$ vertices and can be identified with the $x$-edge connecting them. Similarly we can show that  $Y/Z$ errors 
can be identified with the $y/z$ edges of the triangle. 
\end{proof}
By Lemma~\ref{lm:error2edge} we can obtain the following correspondence between edges in $\Gamma^\ast$ and Pauli errors. 

\begin{corollary}\label{co:edges2error}
Any collection of edges in $\Gamma^\ast$ corresponds to an error on the cubic subsystem color code. 
\end{corollary}

Note that Corollary~\ref{co:edges2error} does not imply that this correspondence is unique. In fact two distinct collections of edges could correspond to the same error. For example, given any face in $\Gamma^\ast$, any pair of edges refer to the same error as the third edge. But these errors will be equivalent up to an element in the gauge group. Therefore, we can still refer to a collection of edges as an error, although strictly speaking it corresponds to an equivalence class of errors.

\begin{lemma}
Any path in $\Gamma^\ast$ connecting a pair of vertices  $u$ and $v$ generates nonzero syndrome on $u$ and $v$ and on no other vertices. Any cycle in $\Gamma^\ast$ generates zero syndrome. Conversely any error with nonzero syndrome on $u$ and $v$ alone must be a path connecting  $u$ to $v$ up to a cycle. 
\end{lemma}
\begin{proof}
Any vertex in $\mathsf{V}(\Gamma^\ast)\setminus \{ u,v\}$ has an even degree with respect to the edges in the path. Therefore all checks except those associated to 
$u$ and $v$ have an even number of qubits in error. Therefore all these vertices will have zero syndrome. The checks on  $u$
and $v$  have an odd number of errors giving rise to a nonzero syndrome. Since all vertices have an even degree with respect to the edge of cycle, all checks  have zero syndrome.

The edges incident on a vertex are (single qubit) errors which anti commute with the stabilizer associated to that vertex. Hence an even 
degree with respect to a set of edges implies that with respect to the errors that have support on that set of edges, the vertices in the path except at the end have zero syndromes. If an error is not a path then $u$ (and/or $v$) must have even degree with respect the edges of the error. But then the syndrome would be zero on $u$ (and/or $v$). Therefore it must be a path up to cycles.
\end{proof}

Any  path could be augmented by a cycle and it will produce the same syndrome as the path. 

\begin{lemma}
Any combination of errors on a cubic subsystem code will always result in an even number of nonzero syndromes.
\label{Le:even_syn}
\end{lemma}
\begin{proof}
If there is no error then all syndromes are zero and we clearly have an even number of nonzero syndromes. 
Suppose that there is only one error $E$ on qubit $u$ where $E\in \{ X, Y, Z \}$.  Then $u$ participates in the syndrome measurement of exactly three faces, $f_r$, $f_b$ and $f_g$ whose stabilizers are given by $\prod_{v \in \partial(f_r)}X_v$, $\prod_{v \in \partial(f_b)}Y_v$ and $\prod_{v \in \partial(f_g)}Z_v$ respectively.  If $E$ is an $X$ error, then the faces $f_b$ and $f_g$ will have nontrivial syndromes; similarly if $E$ is $Y$ or $Z$  there will be two nonzero syndromes. 

Suppose that the hypothesis holds for errors up to $k$ qubits. If there is an additional error, then the following cases can arise. 
\begin{compactenum}[(i)]
\item No checks  are shared with other errors: The number of nonzero syndromes will increase by two. 
\item Exactly one check is shared: Then the check that is not shared will add one additional nonzero syndrome. The shared check 
could add one if it is nonzero or else it could reduce the nonzero syndromes by one if it is nonzero. In either case the nonzero syndromes are even. 
\item Exactly two checks are shared: If these are zero (one), then total number of nonzero syndromes will increase (decrease) by two. 
If only one of them is nonzero, the total number of nonzero syndromes will remain unchanged. 
\end{compactenum}
In all these cases the number of nonzero syndromes is even and by the principle of induction it holds for all errors on the subsystem code. 
\end{proof}

\begin{corollary}
Error estimation from syndrome on $\Gamma^\ast$ is equivalent to finding a collection of paths terminating on the vertices with nonzero syndromes.  
\label{co:syndrome2error}
\end{corollary}
\begin{proof}
By Lemma~\ref{Le:even_syn} there are an even number of nonzero syndromes in $\Gamma^\ast$, say $2m$. We can form $m$ pairs and each pair can be identified with a path up to cycles. 
This collection of paths corresponds to an error that has the given syndrome. 
\end{proof}

We are now ready to give the decoding algorithm. 
\begin{theorem}[Colored matching decoder]\label{th:colored-matching}
The decoder given in Algorithm~\ref{alg:colored-matching} returns the  minimum weight error for a given syndrome
on a cubic subsystem color code. 
\end{theorem}
\begin{proof}
The algorithm constructs a complete graph $\mathcal{K}$ with as many vertices as there are nonzero syndromes. 
The edge between any two vertices in this complete graph are given by the length of the shortest path between them in $\Gamma^\ast$.
We know that every such path corresponds to an error on $\Gamma^\ast$. A matching on $\mathcal{K}$ will correspond to an error with 
the same syndrome as observed. A minimum weight matching will correspond to the error with the smallest weight for this syndrome. 
Matching in $\mathcal{K}$  will lead to an error that has the same syndrome as observed on the  stabilizers of 
trivial homology. To  account for the syndrome on the stabilisers of nontrivial homology we need to proceed as in Algorithm~\ref{alg:ktc-projection2},  lines 14--19. 
\end{proof}

\begin{algorithm}
\caption{{\ensuremath{\mbox{Decoding  cubic subsystem color codes by colored matching}}}}\label{alg:colored-matching}
\begin{algorithmic}[1]
\REQUIRE {A 3-colex $\Gamma$, Syndromes on vertices  $s_v,  v \in \mathsf{V}(\Gamma^{*})$ and $s_i$ on stabilisers of nontrivial homology.}
\ENSURE {$\hat{E}$ with same syndrome as observed.}
\STATE $\mathcal{K}:=$ complete graph on vertices $u$ for $s_u\neq0$
\FOR {vertex $v$ with $s_v\neq0$}
\FOR {  $u$ with $s_u\neq0$  and $u\neq v$}
\STATE  Let $p_{uv}$ be the shortest path between $u$ and $v$ 
\STATE $d(u,v) = \#$  edges in  $p_{uv}$
\STATE Form an edge between $u'$ and $v'$ of  weight $d(u,v)$ 
\ENDFOR
\ENDFOR
\STATE Find a minimum weight matching on $\mathcal{K}$, denote it $\mathsf{E}_m$
\STATE ${\mathsf{P}}:=$ paths corresponding to $\mathsf{E}_m$
\STATE  $\hat{E}:=$ error corresponding to ${\mathsf{P}}$  
\FOR {$1\leq i \leq 2g$}
\STATE Let $s_i'$ be the syndrome of $\hat{E}$ with respect to the  $i$th nonlocal stabilizer 
\STATE Update the syndrome $s_i =s_i\oplus s_i'$
\ENDFOR
\STATE Let $\hat{E}_i$ be a $Z$ error that anticommutes with all stabilizers except $i$th nonlocal stabilizer.
\STATE Return $\hat{E}=\hat{E} \prod_i \hat{E}_i^{s_i}$

\end{algorithmic}
\end{algorithm}

\section{ Decoding  Topological subsystem color codes}\label{sec:TSScs}

Topological subsystem color codes are obtained by vertex expansion of a color code. We study the decoding of these codes and their relations to 
surface codes. First we propose a two step decoder for the topological subsystem color codes. 

Throughout this section we assume that the topological subsystem color code is built from a hypergraph $\mathcal{H}_\Gamma$ which in turn is obtained by vertex expansion of a 2-colex $\Gamma$.

\subsection{A two step decoder for topological subsystem color codes} 

We extend the two step decoder proposed for cubic codes to topological subsystem color codes. The basic idea is to correct $X$ errors first and then the $Z$ errors. The structure of the stabilizer generators then leads to a decoder that is local for $X$ errors but  global for $Z$ errors. 

\begin{lemma}\label{lm:tscc-x-errors}
 An $X$ error on a topological subsystem color code can be corrected up to a gauge operator by 
measuring the $Z$-stabilizers. The correction operator for a face $f$ with nonzero syndrome  is  
$X_v$ for some $v\in f$.
\end{lemma}
\begin{proof}
From the construction of $\mathcal{H}_\Gamma$, we see that every face in $\Gamma$ gives rise to another face in 
$\mathcal{H}_\Gamma$. 
Every face of $\mathcal{H}_\Gamma$ contains exactly one vertex of each rank-3 edges of $\mathcal{H}_\Gamma$. 
Hence, the vertices belonging to each of these faces form a partition of all the vertices in $\mathcal{H}_\Gamma$. 
There are two linearly independent stabilisers that we can associate to each face a rank-2 stabilizer, 
$W_1^f$ and  a rank-3 stabilizer, $W_2^f$. 
By  Eq.~\eqref{eq:rank2-stab} the rank-2 stabilizer is a $Z$-type stabilizer, therefore it cannot detect $Z$-type errors. Furthermore the qubits in $W_1^f$ are disjoint from the qubits in $W_1^{f'}$.  
Therefore, if the  rank-2  stabilizer of $f$ has a nonzero syndrome, then we know that  $f$ has an odd number of $X$
errors. 
An $X$ error on any qubit of $f$ is equivalent to an $X$ error on some other qubit of the face upto $Z$ errors. 
Hence we can correct the $X$ errors by simply applying a correction to each face independently of others. 
\end{proof}
Note that for each face $f$ we can pick any one of its vertices to apply the error correction. 

\begin{center}
\begin{figure}[h!]
\centering
\resizebox{0.3\textwidth}{!}{%
  \includegraphics{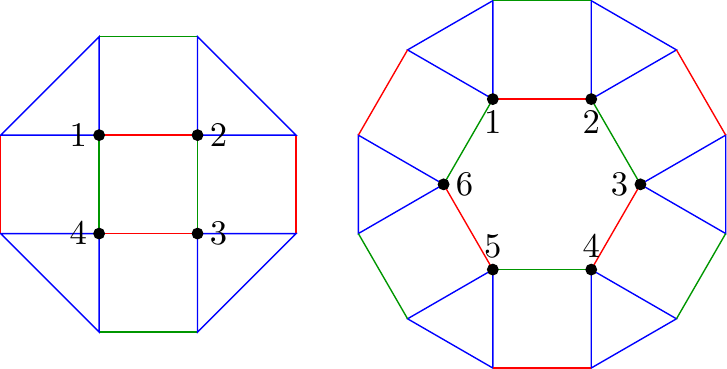}
}
\caption{Correcting $X$ errors on subsystem codes obtained by vertex expansion. If there are an odd number of $X$ errors then the $Z$ type stabilizer associated with the face will register a nonzero syndrome. We can then apply a correction $X_v$
for any one of the qubits in the face. }
\label{lemma_class1}
\end{figure}
\end{center}

Once the syndromes of all the rank-2 stabilizers i.e. the $Z$ type stabilizers  are cleared, the residual error is equivalent to a  $Z$ type error. 
We show that these errors can be corrected by  mapping  onto a color code on $\Gamma$. 

Let us define $\pi$ to be inverse of the vertex expansion of $\Gamma$ so 
that $\pi(\mathcal{H}_\Gamma)=\Gamma$. 
Let $v \in \mathsf{V}(\Gamma)$, then we can associate to $v$ a unique rank-3 edge 
$(v_1,v_2,v_3)\in \mathsf{E}_3(\mathcal{H})$.
  We extend $\pi$ to errors on $\mathcal{H}_\Gamma$ as 
\begin{eqnarray}
\pi(Z_{v_1}) =\pi(Z_{v_2})=\pi(Z_{v_3}) =Z_v \label{eq:vertex-projection}
\end{eqnarray}
It follows that a $Z$ error in the gauge group is mapped to identify under $\pi$.
For instance,  $\pi(Z_{v_1}Z_{v_2})=I$. 
Let $E'$ be an error that is equivalent to $E$. Then we define $\pi(E')= \pi(E)$.
Finally, we project the syndrome of rank-3 stabilizers from $\mathcal{H}_\Gamma$ to $\Gamma$ as follows. Let $s_f$ be the syndrome associated to the rank-3 stabilizer $W_2^f$. 
\begin{eqnarray}
\pi(s_f)  = s_f \label{eq:pi-syn-proj}
\end{eqnarray}
This is well defined because for every face of $\mathcal{H}_\Gamma$, we have a face in $\Gamma$.

\begin{theorem}[Projecting $Z$ errors onto a color code]\label{th:tscc-Z-errors} 
Let $E$ be a 
$Z$-type error or an equivalent error  on a subsystem  code on $\mathcal{H}_\Gamma$. 
Then the syndrome of $\pi(E)$
is identical to the syndrome of $E$.
\end{theorem}
\begin{proof}
Assume that $\Gamma$ has $n$ qubits.
By Lemma~\ref{lm:tscc-params}, the subsystem code has $3n$ qubits, and subsequent to correcting the $X$ errors we have to account for $2^{3n}$ phase errors. 
However, a $Z$ error on any  qubit of a rank-3 edge $(u,v,w)$ are  equivalent i.e.
  $Z_u,Z_v, Z_w, Z_uZ_vZ_w$ are equivalent errors. 
  Thus modulo the gauge group we have to account for only $2^n$ phase flip errors. This is precisely the same number of phase errors on $\Gamma$.

  Let $e=(v_1,v_2,v_3)\in \mathsf{E}_3(\mathcal{H}_\Gamma)$. Then a $Z$ error on  $v_i$ 
causes a nonzero syndrome on the three faces that contain the rank-3 edge. Specifically, the rank-3 stabilizers that contain  $e$. 
These   stabilizers are associated to three distinct faces in $\Gamma$. Now the error $\pi(Z_{v_i}) = Z_v$ has a nonzero syndrome on precisely  the same faces. 
By linearity of $\pi$, the statement holds for all the $Z$ type errors and errors equivalent to $E$ since
they result in the same syndrome as $E$. 
\end{proof}

An immediate consequence of Theorem~\ref{th:tscc-Z-errors} is that  we can 
decode a  $Z$ type error on 
the subsystem color code by  projecting it onto the underlying 2-colex and then lifting it back. 
Thus we have the following decoding algorithm. 

\begin{algorithm}[ht]
\caption{{\ensuremath{\mbox{Decoding topological subsystem color codes}}}}\label{alg:tcc-projection1}
\begin{algorithmic}[1]
\REQUIRE {A  substem color code on a hypergrah $\mathcal{H}_\Gamma$, where $\Gamma$ is a 2-colex and 
syndromes  $s_f^{i},  f \in \mathsf{F}(\mathcal{H}_\Gamma)$ and $1\leq i\leq 2$.}
\ENSURE {Error estimate $\hat{E}$ such that $\hat{E}$ has the input syndrome.}
\STATE $\hat{E}=I$
\FOR {$f\in \mathsf{F}(\mathcal{H}_\Gamma)$} \COMMENT  Correct $X$ errors locally 
\IF {$s_f^{1}\neq 0$} 
\STATE $\hat{E}= \hat{E} X_v$ for some  $v\in f$ \COMMENT  Apply correction on exactly one vertex of $f$
\STATE Flip $s_f^2$ and $s_{f'}^2$ where $f'$ is the face adjacent to $f$ and $\{ W_{f'}^2,X_v \}=0$ 
\ENDIF
\ENDFOR
\STATE Project syndromes $s_f^2$ onto $\Gamma $ for all $f\in \mathsf{F}(\mathcal{H}_\Gamma)$ \COMMENT Correct $Z$ errors globally
\STATE Decode the errors on $\Gamma $ using any 2D color code decoder. Denote by $E'$, the estimate. 
\STATE Let $\Omega$ be the vertices in   $E'$.
\FOR { each vertex $v \in \Omega$} \COMMENT Lift the error to subsystem code
\STATE $\hat{E} = \hat{E} Z_{v_1} Z_{v_2} Z_{v_3}$ where $(v_1,v_2,v_3)$ is the rank-3 edge associated to $v$ 
\COMMENT  It can also be lifted to $\hat{E} = \hat{E} Z_{v_1} $ because $Z_{v_2} Z_{v_3}$ is in gauge group.
\ENDFOR
\STATE Return $\hat{E}$ 
\end{algorithmic}
\end{algorithm}

\begin{theorem}
Subsystem color codes of Lemma~\ref{lm:tscc-params} can be decoded using Algorithm~\ref{alg:tcc-projection1}. 
\end{theorem}
\begin{proof}
By Lemma~\ref{lm:tscc-x-errors} the $X$-type errors can be corrected locally by considering the $Z$ type (rank-2) stabilizer generators. 
Following the correction some syndromes will change. 
Every vertex $v$ is contained in exactly one rank-3 edge. This rank-3 edge is incident on three faces of $\mathcal{H}_\Gamma$, say $f,f'$ and $f''$. 
The error $X_v$ causes non zero syndrome on exactly two rank-3 stabilizer generators: $W_2^f$ and $W_2^{f'}$ 
where  $v\in f$ where $\{X_v, W_{f'}^2\}=0$ while $[X_v, W_{f'}^2]=0$. 
So before we can correct for the $Z$ errors we also update the syndrome corresponding to $f'$.
Following the correction of the $X$ errors, the residual error is equivalent to a $Z$ type error which 
by Theorem~\ref{th:tscc-Z-errors} can be decoded by  projecting onto $\Gamma$. 
\end{proof}

Color codes can be (efficiently) decoded using the methods in \cite{Bombin2012, Delfosse2014, Wang2010, ps2015, ps2017}, which implies that the proposed decoder is also efficient.

\subsection{Uniform rank-3 hypergraph subsystem codes}

In this section we show that the uniform rank-3 subsystem codes obtained from a bicolorable graph can also be decoded using the algorithm in the previous section. 
First we briefly review this construction, the reader is referred to \cite{ps2012} for more details.

\noindent
\hrulefill

\begin{itemize}
\item[] Consider a graph $\Gamma$ such that every vertex has an even degree greater than 2. Construct a 2-colex $\Gamma_2$ from $\Gamma$ using construction 1. Apply construction 2 with $\mathsf{F}$ being the set of $v$-faces of $\Gamma_2$ and the hyperedges on the boundaries of $e$-faces of $\Gamma_2$. 
Denote this hypergraph $\mathcal{H}$. 
\end{itemize}
\noindent
\hrulefill

An illustration of this construction is given in Fig.~\ref{fig:sq_oct_ur3_hsc_const}.
With this construction, we can define a unit cell of $\mathcal{H}$ to be a $v$-face of $\Gamma_2$ along with the hyperedges and the inner face introduced by construction 2. This ensures that the support of two unit cells will not overlap and the support of all the unit cells will cover the entire lattice.

The gauge group is generated by all the link operators of $\mathcal{\overline{H}}$.
Every face of $\mathcal{H}$ gives rise to two linearly independent cycles that are homologically trivial. The associated loop operators generate the stabilizer of the subsystem code on $\mathcal{H}$.

We now show that these subsystem codes can also be obtained by vertex expansion but from a different 2-colex. 
\begin{theorem}
The uniform rank-3 hypergraph subsystem codes obtained from a bicolorable graph $\Gamma$ can also be obtained from the vertex expansion of a 2-colex.
\end{theorem}

\begin{proof}
Let $\Gamma$ be a bicolorable graph, with each vertex having an even degree, 
 from which the hypergraph $\mathcal{H}$ representing the uniform rank-3 hypergraph subsystem code 
is constructed. 
Let $\Gamma_2$ be the intermediate 2-colex from which the hypergraph is constructed. 
(Refer to the section \ref{ssec:hg-subsys-codes}.)
Recall that the 2-colex from $\Gamma$ has three types of faces: $e$-faces which are 4-sided, $f$-faces and $v$-faces. 
Now contract the rank-3 edges added inside the $v$-faces. 
This causes all the rank-3 edges to become vertices and the $e$-faces become parallel edges between two such vertices. 
Replace these parallel edges by a single edge. 
The vertices now become trivalent. 
The vertices of this new graph are exactly the vertices obtained from the rank-3 edges and the edges from the $e$-faces and
the faces from the $v$ and $f$-faces of $\Gamma_2$. 
By assumption $f$-faces are bicolorable, together with the $v$-faces the new graph is a 3-face-colorable
and trivalent. 
Therefore it is a 2-colex. 
But the operations performed are exactly the reverse of the vertex expansion operation. 
In other words, the subsystem code on $\mathcal{H}$  could have been obtained by vertex expansion also.
\end{proof}

From this result it follows that the uniform rank-3 subsystem codes from bicolorable graphs can be decoded by 
Algorithm~\ref{alg:tcc-projection1}.

\section{Generalized Five-Squares Subsystem Codes}\label{sec:NUR3HSSCs}

In this section we study a generalization of the Five squares code proposed in \cite{Suchara2010} .
These codes were  proposed in \cite[Theorem~6]{ps2012}. We call them nonuniform rank-3 subsystem codes.
They also build on the Construction~\ref{alg:hg-tsc-construction}.
Recall that there are three main ingredients to this construction. 
We need to pick  i) a 2-colex $\Gamma_2$ ii) a collection of faces $\mathsf{F} \subset \mathsf{F}_r(\Gamma_2)$
and iii) edges in $\Gamma_2$ which are promoted to rank-3 edges.

\begin{construction}[h]

\caption{{\ensuremath{\mbox{Generalized five-squares subsystem codes \cite{ps2012}}}}}\label{alg:gen-5s-tsc}
\begin{algorithmic}[1]
\REQUIRE {A bipartite graph $\Gamma$  in which all the vertices have even degrees greater than 2.}
\ENSURE { Subsystem code on  a  hypergraph $\mathcal{H}$.}
\STATE Choice of 2-colex:  
Let $\Gamma_m$ be the medial graph of a bicolorable graph$\Gamma$. Using construction 1 on $\Gamma_m^*$, the dual of $\Gamma_m$,
 obtain a three face colorable graph $\Gamma_2$.

\STATE Choice of $\mathsf{F}$:  Let the set of $v$-faces of $\Gamma_2$ be $F_v \cup F_f$ 
where $F_v $ is the set of faces obtained from the vertices of $\Gamma$ and $F_f$ from the faces of $\Gamma$.
Apply construction 2 with $ \mathsf{F} = F_v $ and let 
the hyperedges not be in the boundaries of the $e$-faces of $\Gamma_2$.

\STATE Choice of edges to be promoted: Outgoing edges of $e$-faces that also lie in $\mathsf{F}$. 
Denote  the resulting graph as $\mathcal{H}$ and construct the subsystem code on $\mathcal{H}$.

\end{algorithmic}
\end{construction}

The hypergraph obtained above is 3-edge-colorable hypergraph. 
We choose the following coloring scheme for the edges similar to that in \cite{ps2012}.
This coloring can be derived from the edge coloring of the 2-colex $\Gamma_2$.
Assume that the $v$-faces of $\Gamma_2$ are colored red, $e$-faces blue and
$f$-faces green. This implies that the $e$-faces which are all 4-sided must be bounded by red and green edges. 
Retain the edge coloring of the 2-colex whenever that edge is present in $\Gamma_2$.
Promoting the edges not in the boundaries of the $e$-faces implies that the edges  promoted to rank-3 edges are blue. 
The edges of the faces newly introduced  are colored red and green.
We  summarize this scheme directly with respect to the hypergraph as follows. 
\begin{compactenum}[i)]
\item Rank-3 edges are colored blue. 
\item Unpromoted edges of faces in $\mathsf{F}$ are colored green.
\item Outgoing edges of faces in $\mathsf{F}$ are colored red.
\item Remaining edges of $e$-faces green.
\item Every cycle of new rank-2 edges i.e. in $\mathcal{H}$ but not in $\Gamma_2$, are colored green and red in an alternating manner.
\item Rest of edges are colored blue. 
\end{compactenum}
This particular edge coloring scheme simplifies the decoding procedure we proposed for these subsystem codes. 
 A unit cell for this graph can be a 
face $f \in F_v$ along with the newly introduced hyperedges and the inner face and the surrounding $e$-edges 
of $\Gamma_2$.

The subsystem code is now defined by the resulting hyper graph. The gauge operators and the stabilizers are defined as in 
Section~\ref{ssec:hg-subsys-codes}.
The following five types of cycles generate the stabilizer of the subsystem code, see Fig.~\ref{fig:class2TSC-stabilizers}.
for an illustration:

\begin{enumerate}[(S1)]
\item A cycle consisting of simple edges surrounding a face $f$ where $f \in F_f$. This corresponds to stabilizer of the form 
$\prod_{v\in \sigma} X_v$.

\item A cycle formed by alternating simple and hyperedges on the outer boundary of a face $f \in F_v$ and alternating simple edges on the boundary of its inner face introduced by construction 2. This stabilizer takes the form $\prod_{(u,v,w)\in f}X_uX_vX_w$.
\item A loop consisting of simple and hyperedges around a face $f \in F_f$which  are on the boundaries of the surrounding faces. Let  $\Omega_o$ and $\Omega_i$ be the qubits in the outer and inner boundary of $f$. Then with respect to our coloring scheme, the stabilizer can be written as $\prod_{u\in \Omega_o}Y_u\prod_{v\in \Omega_i}X_v$. 
\item A loop of simple edges forming the inner face of a face $f \in F_v$. 
 This corresponds to stabilizer of the form 
$\prod_{v\in \sigma} Z_v$.
\item A loop of simple edges forming the $e$-faces  of $\Gamma_2$. This corresponds to stabilizer of the form 
$\prod_{v\in \sigma} Z_v$.
\end{enumerate}

\begin{figure*}[t!]
\centering

\begin{subfigure}[t]{0.5\textwidth}
\centering

\includegraphics[width=70mm]{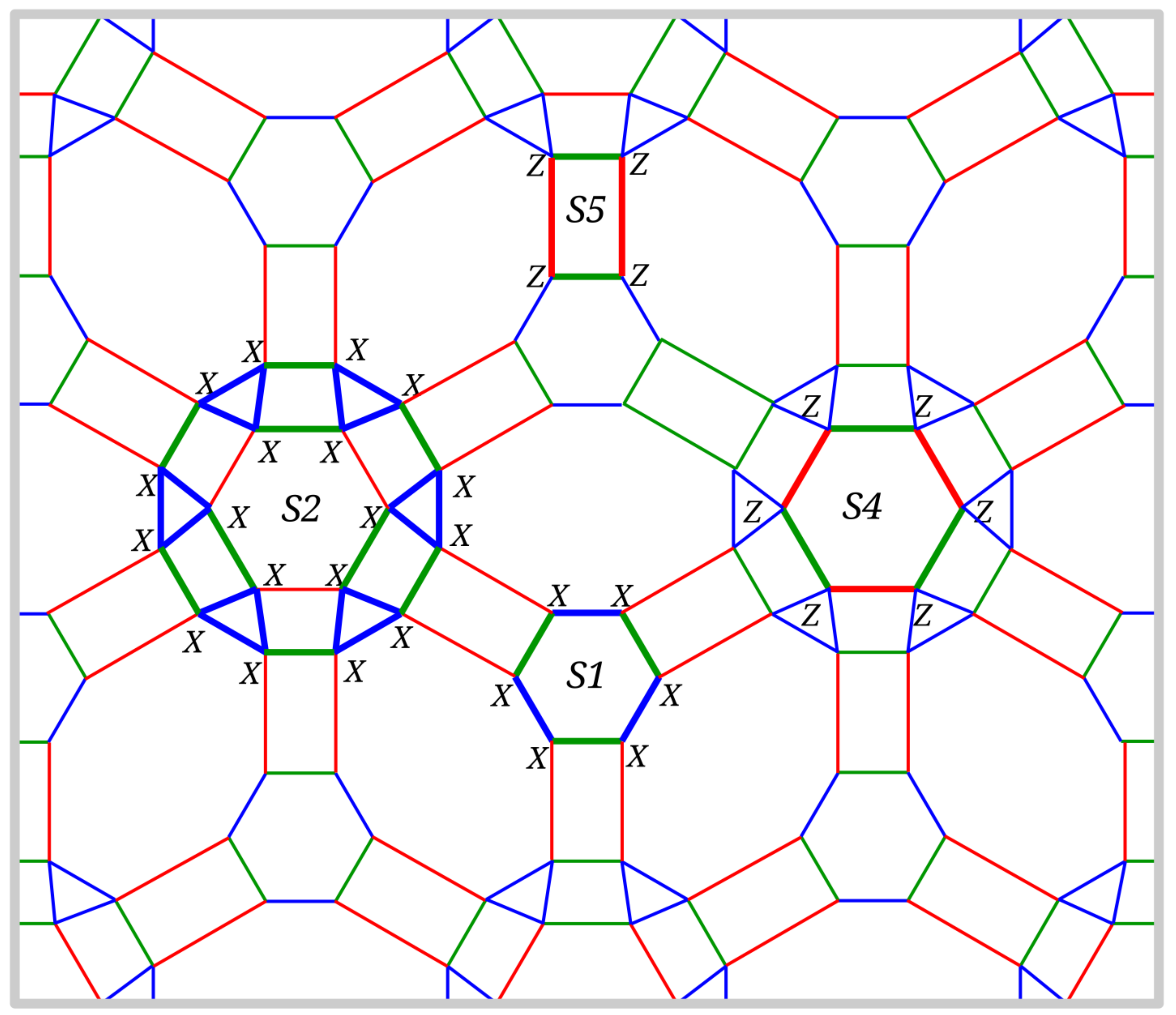}

\captionof{figure}{Stabilizers of type S1, S2, S4 and S5}
\end{subfigure}%
~
\begin{subfigure}[t]{0.5\textwidth}
\centering

\includegraphics[width=70mm]{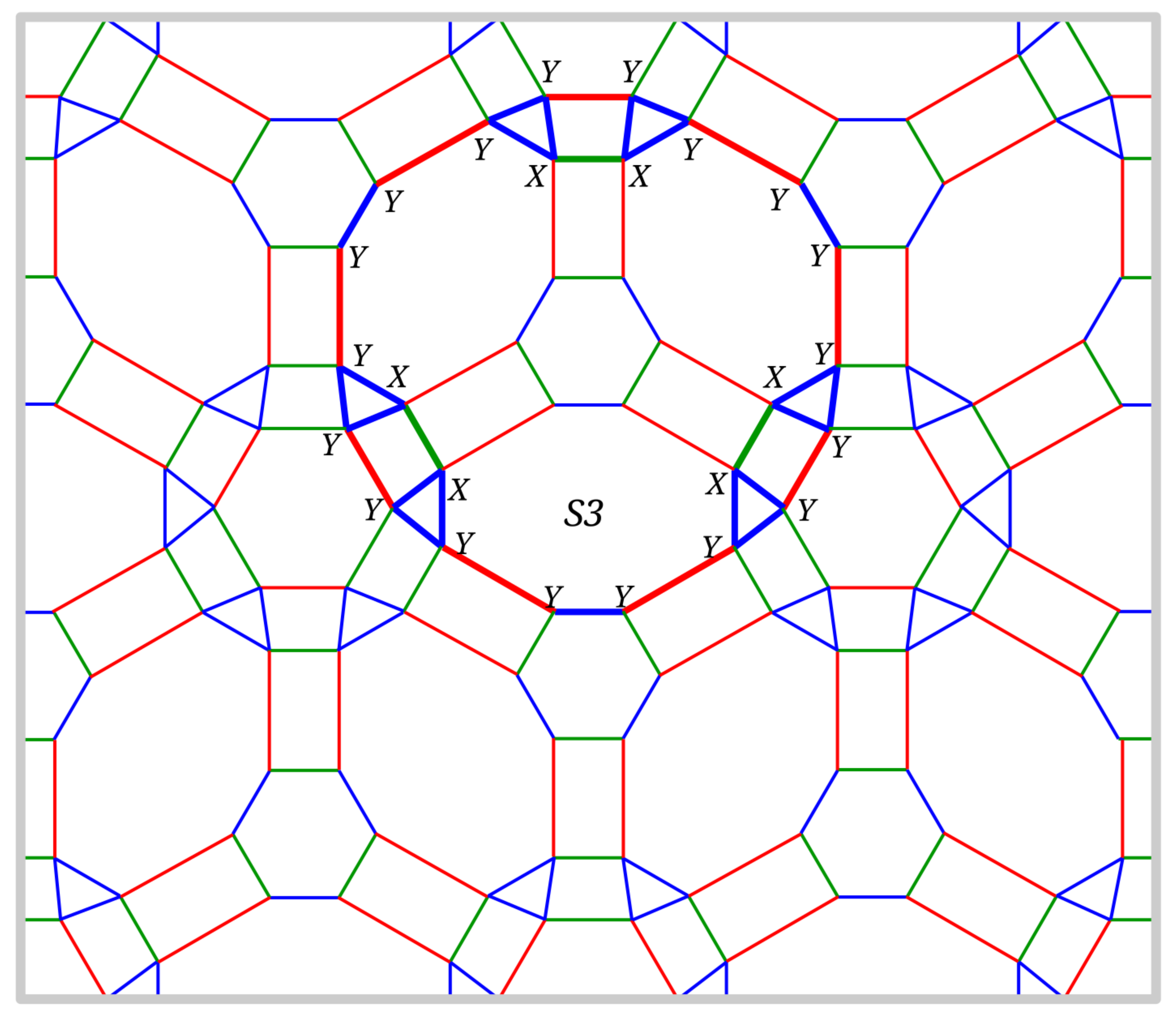}

\captionof{figure}{Stabilizers of type S3}
\end{subfigure}
\caption{(Color online) Stabilizers for nonuniform rank-3 hypergraph subsystem codes}
\label{fig:class2TSC-stabilizers}
\end{figure*}

In the following subsections an algorithm is presented to decode these subsystem codes.

\subsection{Correcting  $X$ errors on generalized five-squares codes}

As in the previous cases we shall first correct for the $X$ errors on the generalized five squares codes. 
This is bit more complex than the algorithm for topological subsystem color codes.
Let $f \in \mathsf{F}$. 
Then we define a unit cell as $f$ and the $e$-faces adjacent to it including the rank-3 edges and the inner face of
$f$. Examples are shown in Fig.~\ref{lemma_class2}.
These unit cells form a partition of the vertices (qubits) of the hypergraph. 
Then we can show the following result.

\begin{lemma}\label{lm:g5s-x-errors}
Bit flip errors on 
the generalized five-squares codes
can be corrected up to a gauge operator by measuring 
$Z$ only stabilizer generators. The residual error is equivalent to a $Z$ only error.  
\end{lemma}
\begin{proof}
In each unit cell associated to a face $f\in \mathsf{F}$, there are $Z$-only stabilizers associated to the $e$-faces and the inner face of $f$ ie stabilizers
of type (S4) and (S5) listed above. 
All these stabilizers have disjoint support. 
In each set of qubits,  as we have already seen by Lemma~\ref{lm:-cubic-reduction}, all single qubit $X$ errors are equivalent 
modulo the gauge group (up to additional $Z$ errors). 
So we only need to correct for one bit flip error for each of these faces. 
The error correction can be applied on any one designated qubit for each of the faces. 
\end{proof}

\subsection{Correcting $Z$ errors on generalized five-squares codes}
After the correction of the bit flip errors modulo the gauge group, we are left with $Z$ only errors or errors equivalent to them. 
Then modulo the gauge group, only a subset of qubits from each unit cell can be considered to correct the $Z$ errors as will be shown in the following lemma.

\begin{lemma}
After all the $X$ errors are corrected on a generalized five-squares code, any combination of $Z$ errors on a unit cell, can be reduced to one  $Z$ error for each $e$-face of the unit cell and exactly one $Z$ error on one rank-3 edge as per the pattern shown in Fig.~\ref{lemma_class2}. The $Z$ error on each $e$-face can be chosen to be on the outer boundary of the unit cell.
Further, the $Z$ error on the rank-3 edge can be corrected by measuring the stabilizer of type S2 supported on the unit cell. 
\end{lemma}
\begin{proof}
Let us consider the set of qubits on a unit cell to be a union of disjoint subsets of qubits residing on the inner face
and the $e$-faces, 
say $e_1,e_2,\dots e_l$ when the number of edges on the inner face is $l$. Each $e$-face consists of four qubits two of which are connected to the inner face by hyperedges. 

Now we make the following observations. Any set of $Z$ errors on a rank-3 edge $(u,v,w)$ can be reduced to identity or a single qubit $Z$ error, say $Z_u$ using the gauge operators $Z_uZ_v$, $Z_vZ_w$, and $Z_wZ_u$.
Let us assume that this reduced error is on a qubit which is common to the rank-3 edge and an $e$-face.
This means we have at most $l$ such $Z$ errors corresponding to one per rank-3 edge.

Consider a single $Z$ error on a rank-3 edge $(u,v,w)$. This edge shares a vertex, say  $u$, with an  $e$-face. This face also shares a vertex $u'$ with another rank-3 edge $(u',v',w')$. 
We can transfer $Z$ error on $(u,v,w)$ to $(u',v',w')$ using the stabilizer of the $e$-face.
 This could add some new $Z$ errors on  the qubits  of $e$-face which are not incident on any rank-3 edge.
By repeated application of this reduction we can reduce a $Z$ error on the rank-3 edges  to a $Z$ error pattern that is supported on exactly one qubit of one rank-3 edge and the remaining qubits of the $e$-faces of the unit cell. 

Then, we can push one of the $Z$ errors on any $e$-face to another unit cell. 
This is possible because the outgoing edges of the $e$-faces carry $Z$-type gauge operators.
This can be done consistently for all the unit cells to obtain the pattern in Fig.~\ref{lemma_class2},
because these qubits belong also to the 2-colex and an inconsistent assignment will happen only if there is an odd cycle. 
But 2-colexes are bipartite and this is not possible.

Finally note that  the stabilizer of type S2 which is supported in the unit cell is of the form $\prod_{(u,v,w)\in f}X_uX_vX_w$. 
It clearly anticommutes with the $Z$ error that is supported on the rank-3 edge and can detect it. 
Since the other qubits do not have any support on the stabilizer, we can correct the any single qubit errors on the rank-3 edges. 
\end{proof}
Figure \ref{lemma_class2} shows two examples of unit cells of nonuniform rank-3 hypergraph subsystem codes with the highlighted qubits being the ones chosen for correcting the $Z$ errors.

\begin{center}
\begin{figure}[h!]
\centering
\resizebox{0.4\textwidth}{!}{
  \includegraphics{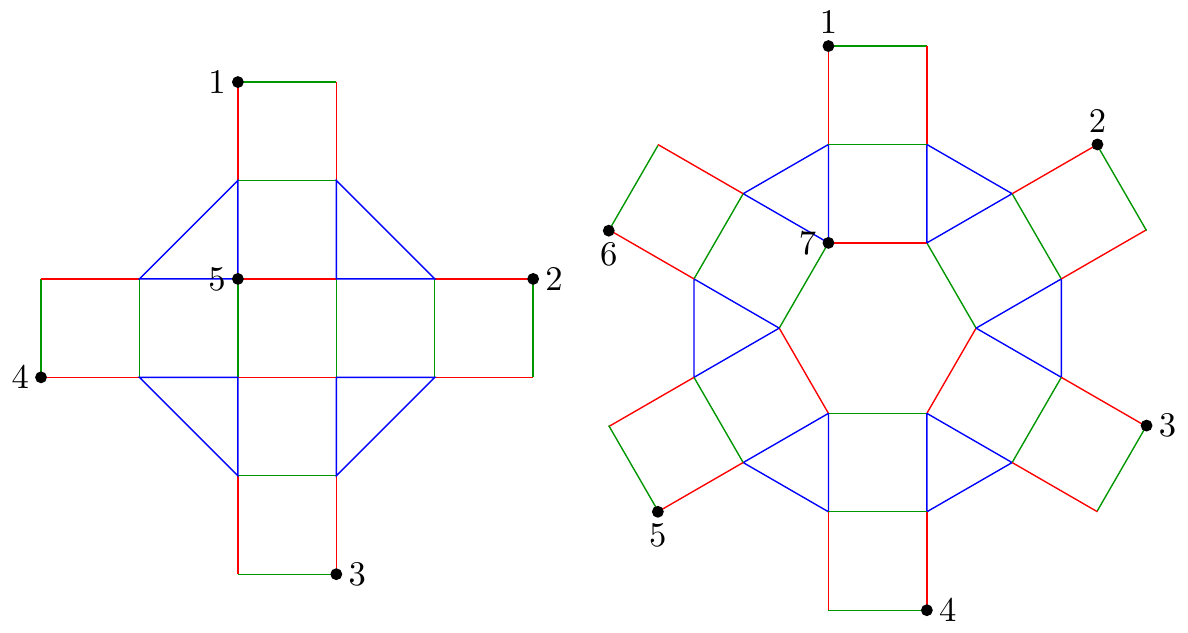}
}
\caption{Examples of Unit cells for nonuniform rank-3 hypergraph subsystem codes. The qubits highlighted are considered for $Z$-error corrections} 
\label{lemma_class2}
\end{figure}
\end{center}

It remains now to correct for the remaining $Z$ errors on the $e$-faces. These errors are detected by the 
stabilizers of type S1 and S3.
We correct these errors by mapping them onto a pair of surface codes.

\begin{lemma}[Mapping to surface codes]
Consider a subsystem code constructed from $\Gamma$ via Construction~\ref{alg:gen-5s-tsc}. Following the correction of $X$ errors and errors on
rank-3 edges, the residual $Z$ errors can be corrected by mapping to a pair of surface codes.
\end{lemma}

\begin{proof}
The residual $Z$-errors on each unit cell can be detected by stabilizers of type S1 and S3. 
Observe that any single $Z$ error on an unit cell
 flips exactly one S1 type stabilizer and one S3 type stabilizer. 
These two stabilizers belong to two distinct faces $f,f'$ in $F_f$ of $\Gamma_2$.
From Consruction~\ref{alg:gen-5s-tsc}, we see that these faces correspond to adjacent  vertices of $\Gamma^\ast$.
We can therefore project these syndromes to the vertices  on $\Gamma^\ast$ and the error to the edge shared by them. 
In this manner we obtain an error pattern on a surface code on $\Gamma^\ast$. (The stabilizers will  be
 mapped to the vertex operators of $\Gamma^\ast$.) 

Observe that  if a qubit participates in the  S1 type check 
on a face $f$, then it does not participate in the check  of the type S3 on the same face and vice versa. 
Hence we can partition $F_f$ into 
two subsets $F_a$, $F_b$ where one set of qubits which participate in S1 stabilizers of $F_a$ and S3 stabilizers of $F_b$
while the other set of qubits participates in S3 stabilizers of $F_a$ and S1 stabilizers of $F_b$.
Thus we can project onto two copies of surfaces codes on $\Gamma^\ast$.
\end{proof}

With these results we obtain the complete decoding algorithm for the generalized five-squares codes
as given in Algorithm~\ref{alg:dec-gen-5s} extending the algorithm proposed for five squares codes in \cite{Suchara2010}.
The proof follows by putting together all the results proved in this section.  We omit the details. 
Efficiency of the decoding algorithm follows from the existence of efficient decoders for  surface codes.

\begin{theorem}\label{th:dec-gen-5s}
Generalized five-squares codes can be efficiently decoded using Algorithm~\ref{alg:dec-gen-5s}.
\end{theorem}

\begin{algorithm}[ht]
\caption{{\ensuremath{\mbox{Decoding generalized five-squares subsystem codes }}}}\label{alg:dec-gen-5s}
\begin{algorithmic}[1]
\REQUIRE {A  generalized five squares substem color code on a hypergrah $\mathcal{H}_\Gamma$, where $\Gamma$ is the bicolorable graph used to construct $\mathcal{H}_\Gamma$ via Construction~\ref{alg:gen-5s-tsc}. }
\ENSURE {Error estimate $\hat{E}$.}
\STATE $\hat{E}=I$
\FOR { each $e$-face  and $f$-face of $\mathcal{H}_\Gamma$} \COMMENT Correct $X$ errors locally 
\STATE Measure rank-2  stabilizer of the face ie stabilizer of type S4 or S5. Let the syndrome be  $s$
\STATE $\hat{E} = \hat{E} X_v^{s}$
\ENDFOR
\FOR {each $f\in v$-face of $ \mathcal{H}_\Gamma$} \COMMENT Correct some $Z$ errors locally
\STATE Measure rank-3 stabilizer of the face ie stabilizer of type S2. Let the syndrome be $s$.
\STATE $\hat{E} = \hat{E} Z_v$ for some $v\in f$.
\ENDFOR
\FOR {for $f$ in $f$-faces of $\mathcal{H}_{\Gamma}$} \COMMENT Correct remaining $Z$ errors globally
\STATE Measure S1 and S3 stabilizer.
\ENDFOR
\STATE Project the S1 syndromes of $f\in F_a$ and S3 syndromes of $f\in F_b$ on one copy of $\Gamma^\ast$.
\STATE Project the S3 syndromes of $f\in F_a$ and S1 syndromes of $f\in F_b$ on one copy of $\Gamma^\ast$.
\STATE Decode the errors on each copy of $\Gamma $ using any 2D surface code decoder. Denote by $E_1$, $E_2$
the estimates on each of copies. 
\STATE Each qubit on the surface codes corresponds to a unique qubit on the $\mathcal{H}_{\Gamma}$.
Denote this set of qubits by $\Omega = \text{supp}(E_1)
\cup \text{supp}(E_2)$.
\STATE Return $\hat{E}= \hat{E} \prod_{v\in \Omega} Z_v$ 
\end{algorithmic}
\end{algorithm}

\begin{figure*}[h!]
\centering
\begin{subfigure}[t]{\textwidth}
\centering

\includegraphics[width=130mm]{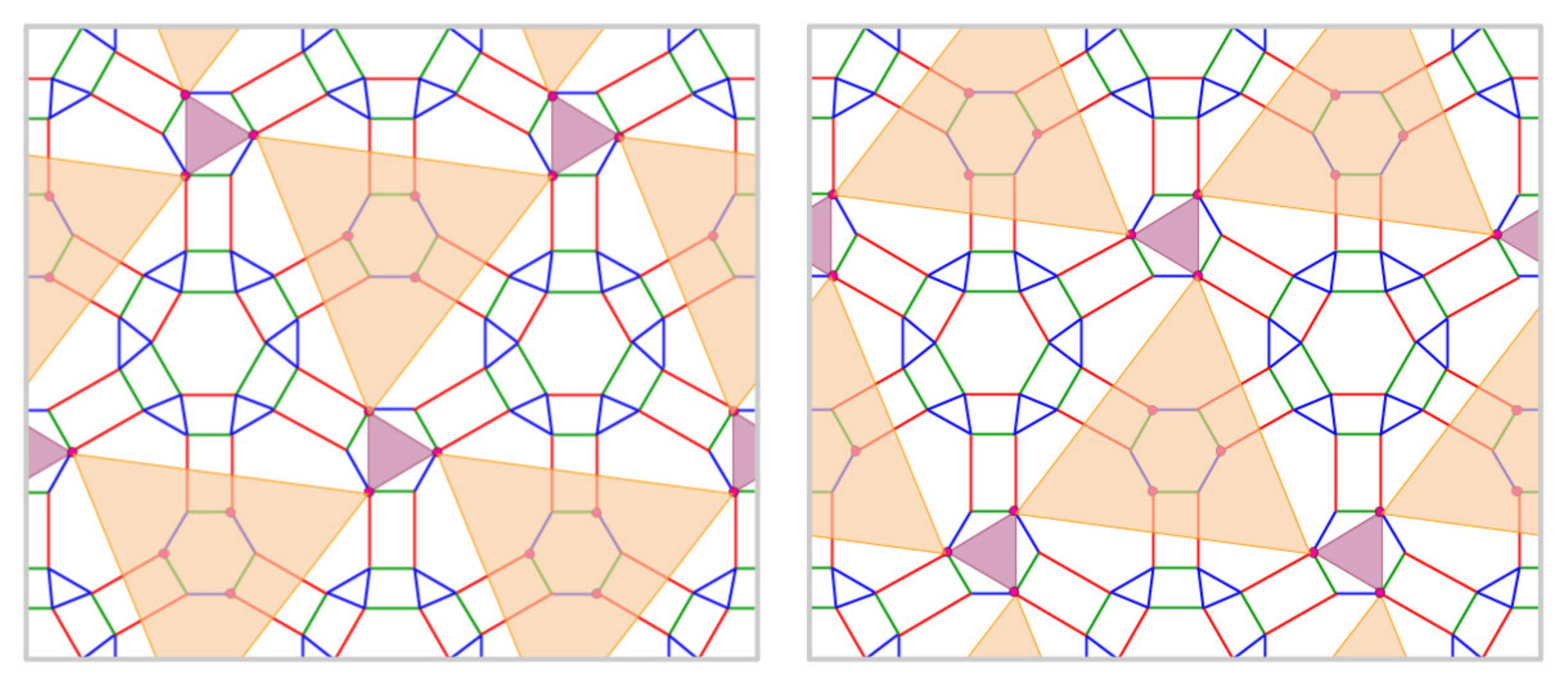}

\captionof{figure}{Mapping to two surface codes}
\label{fig:class2TSC-mapping5sq2surfcodes_1}
\end{subfigure}

\begin{subfigure}[t]{\textwidth}
\centering

\includegraphics[width=130mm]{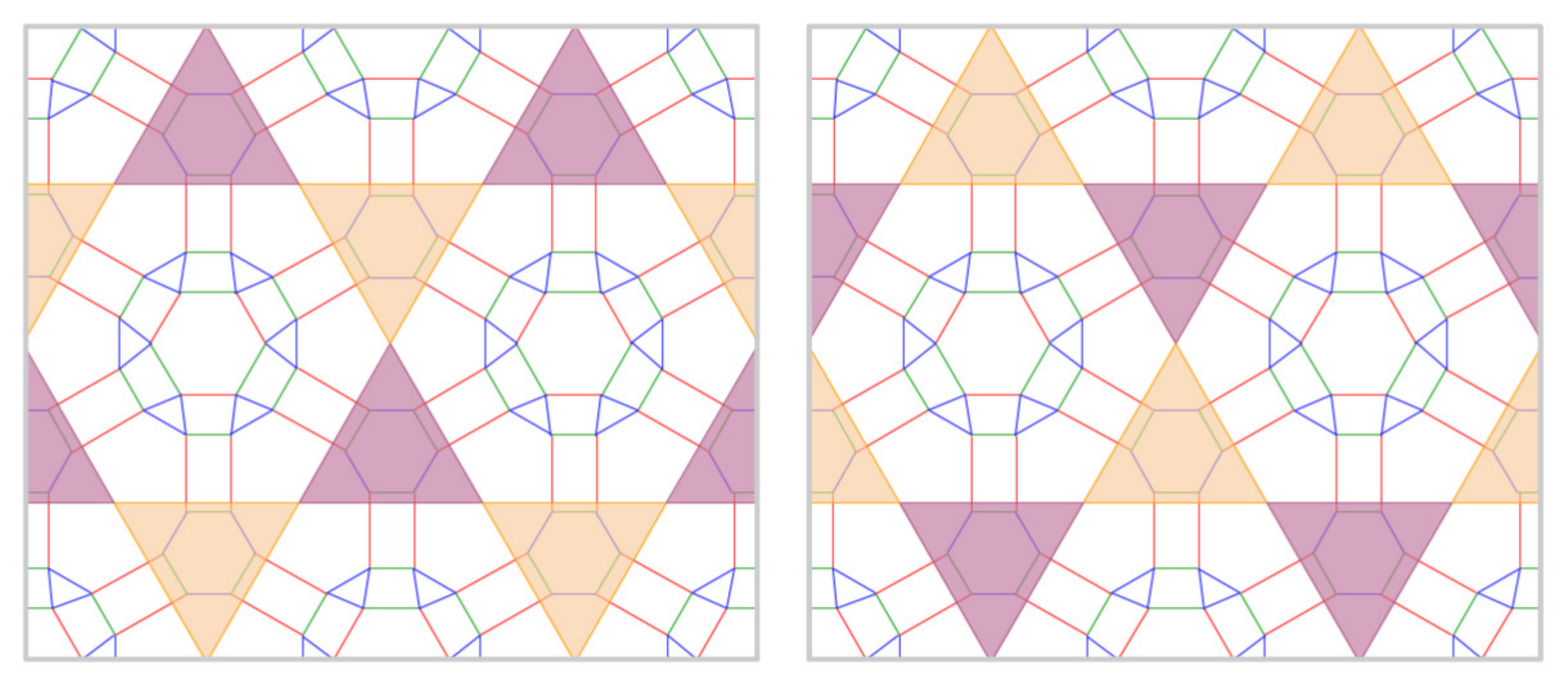}

\captionof{figure}{Redrawing the surface codes}
\label{fig:class2TSC-mapping5sq2surfcodes_2}
\end{subfigure}

\begin{subfigure}[t]{\textwidth}
\centering

\includegraphics[width=130mm]{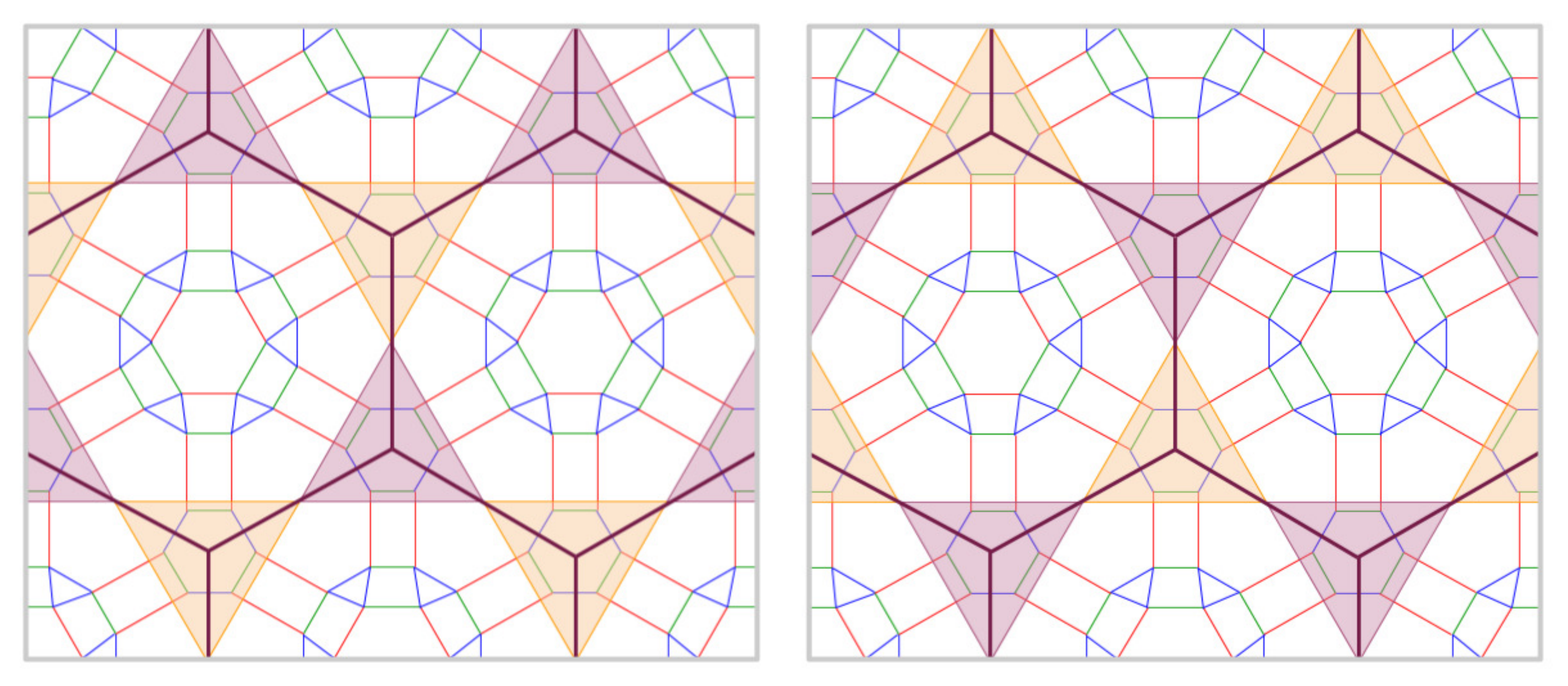}

\captionof{figure}{Representing each qubit by an edge instead of vertex will result in a surface code based on $\Gamma^*$ as shown in this figure by the maroon colored solid edges}
\label{fig:class2TSC-mapping5sq2surfcodes_3}
\end{subfigure}

\caption{(Color online) Example of mapping the generalized five-squares subsystem code to two surface codes constructed from the dual graph of $\Gamma$}
\label{fig:class2TSC-mapping5sq2surfcodes}
\end{figure*}

Fig.~\ref{fig:class2TSC-mapping5sq2surfcodes} illustrates an example of mapping the generalized five-squares subsystem code to two copies of a surface code. This is the same hypergraph subsystem code
in Fig.~\ref{fig:class2TSC-stabilizers} which is constructed from graph $\Gamma$ based on triangular lattice.
In Fig.~\ref{fig:class2TSC-mapping5sq2surfcodes_1} the highlighted qubits on which we correct the $Z$-errors are mapped to two copies of a surface code formed by triangles. The triangles in 
Fig.~\ref{fig:class2TSC-mapping5sq2surfcodes_1} are resized, the orientations are changed and redrawn as in the Fig.~\ref{fig:class2TSC-mapping5sq2surfcodes_2}. The triangles in 
Fig.~\ref{fig:class2TSC-mapping5sq2surfcodes_2} represent a suface code which can also be represented as shown in Fig.~\ref{fig:class2TSC-mapping5sq2surfcodes_2} by the dark colored edges
which is the dual graph of $\Gamma$.

\subsection{A structural result}
Before, we end this section, we prove a structural result on the hypergraph subsytem codes obtained by a variation on  
Construction~\ref{alg:gen-5s-tsc}.  
 Recall that in this construction the rank-3 edges are not in the boundaries of the
$e$-faces. 
We consider a variation where the rank-3 edges are in the boundary of $e$-faces. 
These codes are not good in the sense of having a good distance, but they give an insight into some of the structures that must be avoided in order to obtain good hypergraph subsystem codes.

\begin{theorem}\label{th:tsc-low-d}
Consider a subsytem code obtained via Construction~\ref{alg:gen-5s-tsc}, where in  step 3  the hyperedges are inserted on the $e$-faces of  the 2-colex. Then the resulting subsystem code has $O(1)$ distance independent of the length of the code. 
\end{theorem}

\begin{proof}
The hypergraph from this construction contains unit cells of the form shown in Figure \ref{lemma_bad_code}.
\begin{center}
\begin{figure}[H]
\centering
\resizebox{0.4\textwidth}{!}{
  \includegraphics{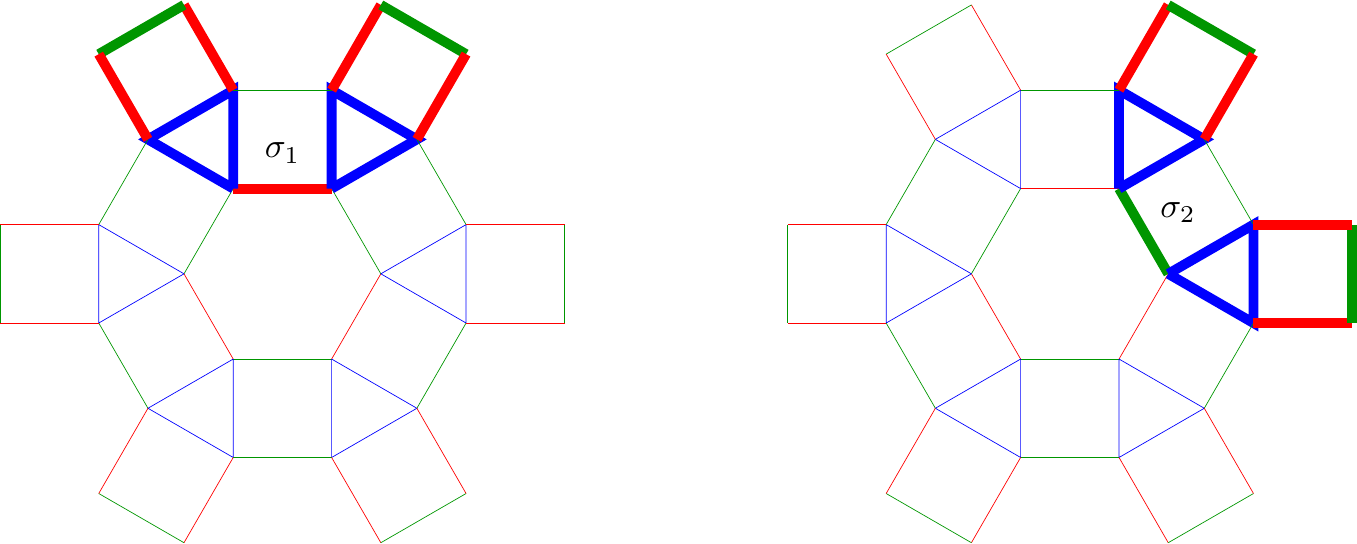}
}
\caption{(Color online) Examples of logical operators with $O(1)$ weight independent of the size of the code.}
\label{lemma_bad_code}
\end{figure}
\end{center}
We can easily check that this hypergraph supports a hypercycle consisting of two neighboring hyperedges and
the adjacent $e$-faces along with a rank-2 edge connecting the hyperedges. (Every vertex in the support of these edges
has degree two with respect to these edges.) The associated loop operator $W(\sigma)$ is in the centralizer of the gauge group. 
Observe that we can also pair a hyperedge with another neighboring hyperedge to give another hypercycle. 
It can be easily verified that the operators associated to these two hypercycles anticommute. Therefore they cannot be in the
center of the gauge group and hence there are not stabilizers of the subsystem code. In other words they are logical operators
with a finite weight. Thus the distance of the code is $O(1)$. 
\end{proof}
In the preceding theorem we assume that the code is sufficiently big enough to contain structures of the form shown in 
Fig.~\ref{lemma_bad_code}. 
Although we only showed that these structures are to be avoided in the context of Construction~\ref{alg:gen-5s-tsc},
these structures should be avoided in any hypergraph subsystem code.

\section{Subsystem surface codes}\label{sec:SSSCs}

The subsystem codes studied so far have been based on hypergraphs
where the gauge group was generated by two body operators. 
In this section we study subsystem codes based on surfaces.
These codes were proposed in \cite{Bravyi2012}.
In these codes the gauge group is generated by weight three generators. 
These codes are amenable for a planar implementation and enable fault tolerant quantum computation
through code deformation techniques unlike the hypergraph subsystem codes which do not allow 
for code deformation through introduction of boundaries.

Recently, a new topological subsystem code based on a surface code was proposed by 
Bravyi et al. These codes have 3-qubit gauge operators and they enable fault tolerant quantum computing  
through code deformation techniques, a feature that was absent in the previously known classes of subsystem 
codes. We generalize this construction to a large class of surface codes. Consequently, we 
are able to obtain a code which has a lower overhead compared to the codes \cite{Suchara2010}.

\subsection{Constructions of subsystem surface codes}
In this section we propose a general method to construct subsystem surface codes. 
Our method enables us to construct codes of lower overhead. 

\renewcommand{\algorithmicrequire}{\textbf{Input:}}
\renewcommand{\algorithmicensure}{\textbf{Output:}}
\begin{construction}[h]
\caption{{\ensuremath{\mbox{ Generalized subsystem surface codes}}}}\label{proc:tssc}
\begin{algorithmic}[1]
\REQUIRE {A 4-valent graph $\Gamma$ embedded on a surface such that all faces are even.}
\ENSURE {A subsystem surface code on a graph $\mathcal{H}$.}
\STATE Find the medial graph of $\Gamma$. Denote it by $\Gamma_m$. Each 4-valent vertex gives rise to a new face. 
Let us call such a face a $v$-face.
\begin{center}
  \includegraphics{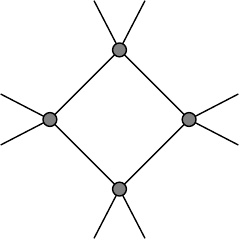}
\end{center}

\STATE In each $v$-face, add a new vertex and connect it to all the vertices in the 
face. Denote this as $\mathcal{H}_\Gamma$. 
\begin{center}
  \includegraphics{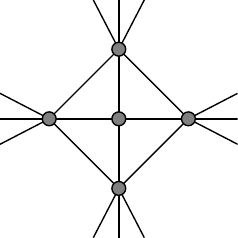}
\end{center}

\STATE Color the resulting complex such that the $f$-faces are all of the same color
and the triangles in the faces are 2-colorable.

\begin{center}
  \includegraphics{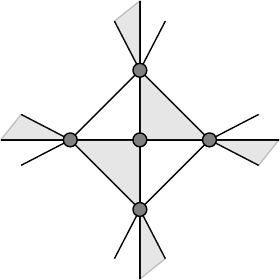}
\end{center}
\STATE Assign a gauge operator $Z_uZ_vZ_w$ or
$X_uX_vX_w$ depending on the color of the triangular face. 

\STATE Let $\mc{G} = \langle g_e \mid e \mbox{  is a triangle in }\mathcal{H}_{\Gamma} \rangle $.
\end{algorithmic}
\end{construction}

\begin{theorem}\label{th:ssc-new}
Let $\Gamma$ be a 4-valent graph embedded on closed surface of genus $g$ such that every face is even sided. 
Then Construction~\ref{proc:tssc}
gives a
subsystem surface code encoding $2g$ qubits into $3v$ qubits,   where $v=|V(\Gamma)|$ and $g$ is the genus of surface.
The code has $v-2g-2$ gauge qubits.
\end{theorem}
\begin{proof}
The number of vertices present in the medial graph $\Gamma_m$ is equal to the number of edges present in $\Gamma$ which is equal to $2v$ owing to the 4-valent nature of $\Gamma$. In addition, after Construction~\ref{proc:tssc}, $v$ vertices are added. So the total number of qubits in the resulting  subsystem surface code is $3v$. 

The construction introduces four gauge operators for each vertex in $\Gamma$. Hence we have $4v$ gauge operators.
Since the product of all the $X$ or $Z$ type gauge operators is identity, we have  $4v - 2$ are independent generators for the
gauge group.

From each face of $\Gamma$, we obtain two stabilizers $s_f^X$ and $s_f^Z$, where 
\begin{eqnarray}
s_f^\sigma = \prod_{e \in \partial f } g_e^\sigma
\end{eqnarray}
where $\partial f$, is the boundary of $f$.
Since $s_f^\sigma $ is generated by the gauge operators, it is in the gauge group. 
We can also check that $s_f^\sigma$ commutes with all the gauge operators and therefore it is in $C(\mathcal{G})$. 
It is evident that $s_f^\sigma $ commutes with the gauge operators of the same type ie $g_e^\sigma$.
Suppose $g_e^{\sigma'}$ is such that $\sigma\neq \sigma'$. 
If $e\in \partial f$, then $g_e^{\sigma'}$ and $s_f^\sigma$ overlap in exactly two vertices and they commute. 
If $e\not\in \partial f$, then it must overlap with $s_f^\sigma$ in at most two  vertices. If it overlaps in two or zero locations,  then again  they commute. 
Gauge operators $g_e$ which overlap exactly once with $s_f^\sigma$ must be in $\mathcal{G}_\sigma$, and hence  commute.
Observe that 
\begin{eqnarray}
\prod_f s_f^\sigma = I.
\end{eqnarray}
Thus, the number of independent stabilizer generators is $2(f-1)$ where $f$ is the number of faces in $\Gamma$. But in $\Gamma$, $v - 2v + f = 2 - 2g$ and therefore $f = 2 + v - 2g$. So the number of independent stabilizer generators is $s=2(v - 2g + 1)$. If $r$ is the number of gauge qubits, $2r =  4v - 2 - 2(v - 2g + 1) $ which gives $r = v + 2g - 2$. The dimension of the subsystem code is $k = n - s - r = 2g$. 
\end{proof}

Note that the faces of $\Gamma_m$ correspond to the vertices and faces of $\Gamma$. 
A face in $\Gamma_m$ which is derived form a face (vertex) in $\Gamma$ is called an $f$-face ( $v$-face). 
In  $\mathcal{H}_\Gamma$, the triangles come from the $v$-faces of $\Gamma_m$ and each $v$-faces gives rise to 
four such triangles. Of more importance are the faces in $\mathcal{H}_\Gamma$ which are derived from the 
$f$-faces of $\Gamma_m$. 
We shall call these faces of $\mathcal{H}_\Gamma$ also as $f$-faces. 
Each of these correspond to faces in $\Gamma$ so without ambiguity we can label them by the faces in $\Gamma$.

If $\chi = 2 - 2g$, the dimension of the code is $2 - \chi$ and the number of gauge qubits is $v - \chi$.
Note that  $\Gamma$ gives a  $[[2v, 2g]]$ surface code. The subsystem code  
adds   50\% overhead. 
The natural question is if we can lower these overheads. 
On the torus, Kitaev's toric code on the square lattice has the parameters $[[2n^2,2]]$. 
This code can be improved to give lower overheads, see \cite{bombin06x}.
Motivated by this construction we 
propose a class of subsystem surface codes that have lower overheads.
We consider a different tiling of the torus proposed in \cite{bombin06x}, see Fig.~\ref{fig:opt-surface-code}.
We observe that this graph meets the conditions of Construction~\ref{proc:tssc}  is applicable. 
Then by Theorem~\ref{th:ssc-new}, we can obtain 
a $[[3(d^2+1)/2,2, (d^2+1)/2,d]]$ subsystem code.
With respect to the $[[3d^2,2,d^2,d]]$ subsystem codes of Bravyi et al., they offer a reduced overhead. 
We conjecture that these are optimal. 

\begin{center}
\begin{figure}[htb]
\centering
  \includegraphics{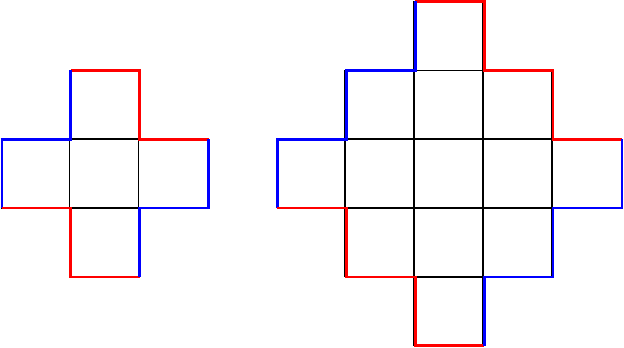}
\caption{(Color online) A  $[[d^2+1,2,d]]$ surface code on a torus for $d=3,5$. Boundaries of same color are to be identified. }\label{fig:opt-surface-code}
\label{fig-opt-surface-code}
\end{figure}
\end{center}

\subsection{Decoding subsytem surface codes}
\begin{theorem}
The subsystem surface codes of Theorem~\ref{th:ssc-new} can be efficiently decoded using Algorithm~\ref{alg:surf-ss-code-decode}.
\end{theorem}
\begin{proof}
As can be seen from the proof of Theorem~\ref{th:ssc-new}, there are 
two types of stabilizers. 
Since they are either $X$-type or $Z$-type we can correct the bit flip and phase flip errors separately. 
Suppose we have a bit flip error on a qubit in the boundary of an $f$-face, such as qubit $v$ in line 2 of Algorithm~\ref{alg:surf-ss-code-decode}. 
This leads non zero syndrome for  exactly two stabilizer generators $s_{f_1}^Z$ and $s_{f_2}^Z$. 
Similarly, bit flip errors on qubits $w$, cause nonzero syndrome  for
$s_{f_2}^Z$ and $s_{f_3}^Z$,   $X_s$ causes nonzero syndrome on $s_{f_3}^Z$ and $s_{f_4}^Z$ and $X_t$ 
on  $s_{f_1}^Z$ and $s_{f_4}^Z$.
A bit flip error on the qubit in the interior of a $v$-face, such as $u$ causes nonzero syndrome with respect
to the $s_{f_2}^Z$ and $s_{f_4}^Z$. 
Note that the stabilizers $s_{f_1}^Z$ and $s_{f_2}^Z$ do not have any overlap with $X_u$ and are unaffected.
We can capture all this information by representing i) each $f$-face by a vertex ii)  each qubit by a edge 
and iii) associate the stabilizer $s_f^Z$ to the corresponding vertex,  as shown in line 2 of
Algorithm~\ref{alg:surf-ss-code-decode}.  
With this mapping we can project the syndromes on the subsystem surface code to the vertices of $\Gamma_X$.
We can then decode the errors on the surface code and lift it back to the subsystem code because there is a one to 
one correspondence between the errors on the surface code and the subsystem code. 
This allows us to lift the errors from the surface code to the subsystem code unambiguously. 
A similar reasoning for phase flip errors leads to the mapping shown in line 3.
We omit the details. 
\end{proof}

\begin{algorithm}[hb!]
\caption{{\ensuremath{\mbox{Decoding subsystem surface code}}}}\label{alg:surf-ss-code-decode}
\begin{algorithmic}[1]
\REQUIRE {A  subsystem surface code on $\mathcal{H}_\Gamma$ constructed from a graph $\Gamma$ according to Construction~\ref{proc:tssc} and syndromes  $s_f^{i},  f \in \mathsf{F}(\Gamma)$ and $i \in \{X, Z\}$.}
\ENSURE {Error estimate $\hat{E}$ such that $\hat{E}$ has the input syndrome.}
\STATE $\hat{E}=I$
\STATE Construct $\Gamma_X$ with vertices labeled by $f$-faces of $\mathcal{H}_\Gamma$ as shown below. 

\begin{center}
	\includegraphics{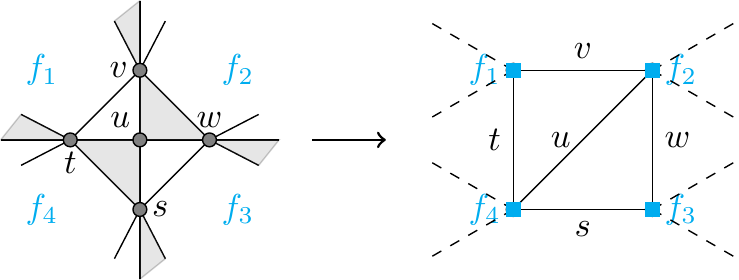}
\end{center}

\STATE Construct $\Gamma_Z$ with vertices labeled by $f$-faces of $\mathcal{H}_\Gamma$ as shown below. 
\begin{center}
	\includegraphics{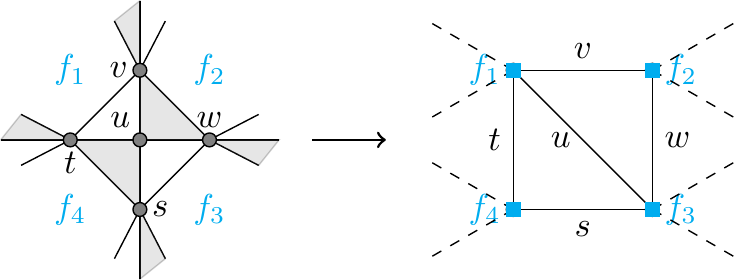}
\end{center}

\STATE Project syndromes $s_f^X$ and $s_f^Z$ onto vertices of $\Gamma_X$ and $\Gamma_Z$ for all $f\in \mathsf{F}(\Gamma)$ where $\Gamma_X$ and $\Gamma_Z$ are constructed from $\mathcal{H}_\Gamma$ as above. 
\STATE Decode the errors on $\Gamma_X$ and $\Gamma_Z$ using any 2D surface code decoder. Denote by $E'_X$ and $E'_Z$, the estimates. 
\STATE Lift the errors $E_X'$, $E_Z'$ to subsystem code. Denote them  $\bar{E}_X$ and $\bar{E}_Z$ respectively.
\STATE Return $\hat{E}=\bar{E}_X \bar{E}_Z$
\end{algorithmic}
\end{algorithm}

\section{Simulation Results}
\label{sec:results}

\begin{figure*}[ht!]
\centering
\includegraphics[width=90mm]{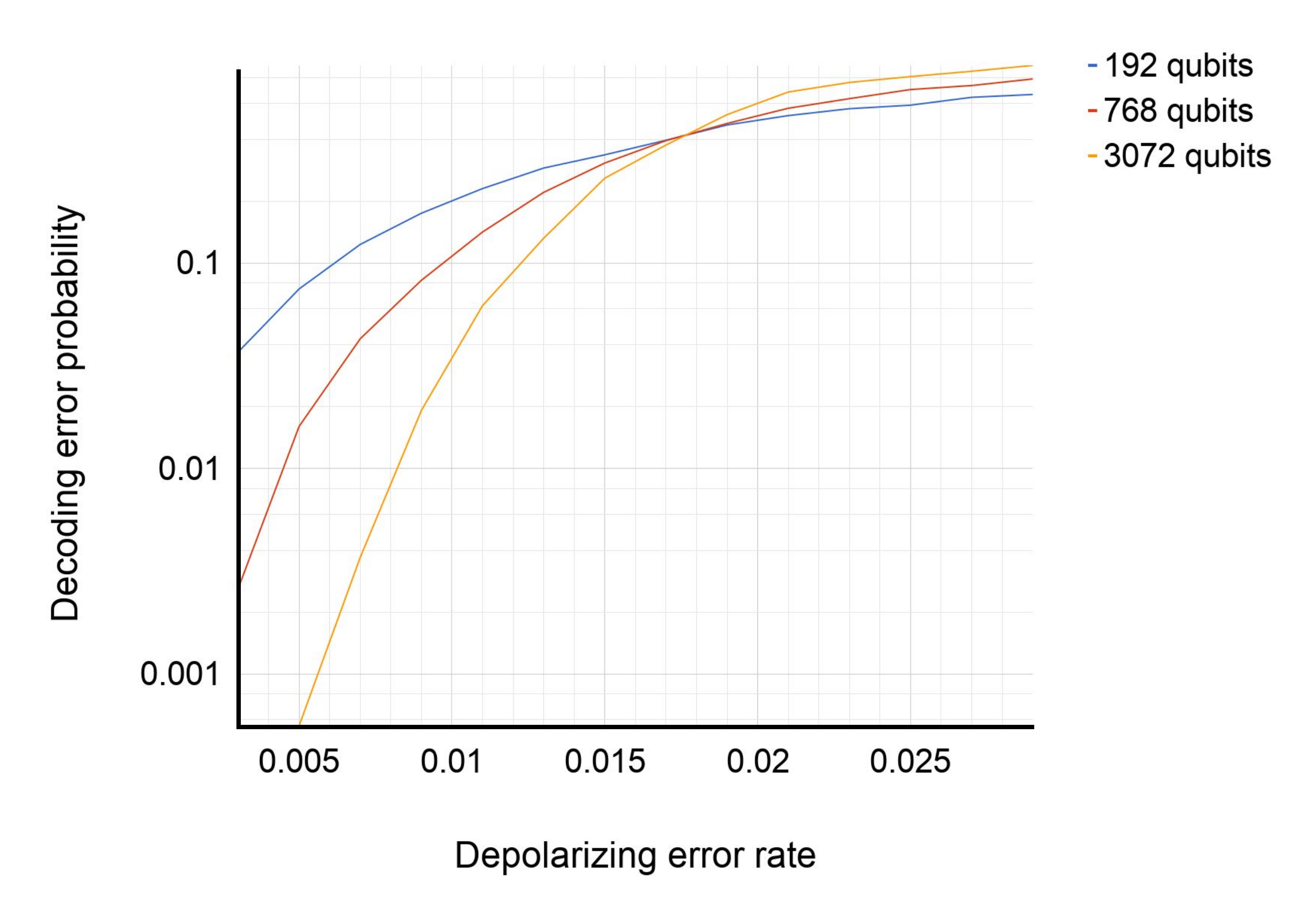}
\captionof{figure}{Performance of the two step decoding algorithm for the hypergraph subsystem code on the square octagon lattice. Noise threshold is about $1.75\%$.}
\label{fig:res-allX-cleanup}
\end{figure*}

We have performed simulations for decoding the subsystem codes obtained from vertex expansion of color codes on square octagon 
lattice, for lattices of different sizes and for different depolarizing error rates. 
We generate a large number of error samples for each error rate and each lattice. For each error sample $E$, we then find an error estimate $\hat{E}$ using one of our decoding algorithms, (specifically Algorithm~\ref{alg:tcc-projection1}). 
If the effective error $E_{eff} = E\hat{E}$ anti-commutes with any of the logical operators of the code, 
then we increment the number of 
failed decoding attempts. Thus obtaining the probability of failure for different error rates, we plot probability of failure against 
the depolarizing noise rate for lattices of different sizes. 

We simulated for depolarizing noise rates ranging from 0 to 0.029 with a step size of 0.002 and for lattices with 192, 768 and 3072 qubits. For each data point 2500 samples were considered.  Fig. \ref{fig:res-allX-cleanup}
shows the result of the simulation. We obtained a noise threshold of about $1.75\%$ for the algorithm using the two step decoding algorithm. 
involving the cleaning up of $X$-errors on each unit cell.
 This is comparable to the threshold of 2\% obtained in \cite{Bombin2012}.
 Since our algorithms do not exploit the correlations between the $X$ and $Z$ errors, it is possible to improve our decoders. 
For the five squares code Bravyi et al. \cite{Suchara2010} obtained a noise threshold of about $2\%$.

\section{Conclusion}
\label{sec:conclusion}

Fault tolerant systems are a necessity without which no quantum computing systems can become a reality. 
Quantum codes are essential for quantum fault tolerance. 
For a quantum code to be useful it is important to have a low complexity decoder. 
Since subsystem codes require only two or three qubit measurements they could be more amenable for experimental realizations. 

In our work, we have proposed algorithms for decoding different families of subsystem codes such as cubic subsystem color codes, topological subsystem color codes, hypergraph subsystem codes and subsystem surface codes.
These decoding algorithms are applicable to a large classes of codes without requiring individual optimization for each code. 
We have shown how to decode cubic subsystem color codes  by mapping them onto a surface code after a preprocessing step. Similarly subsystem surface codes can be decoded by mapping onto other topological codes such as color codes and surface codes.
The advantage of this approach is that there are well established algorithms for decoding surface and color codes. 
We evaluated our decoding algorithms by numerical  simulations and show that they achieve performance comparable to previously known decoders. 

In addition we showed some structural results on subsystem codes which are of independent interest. 
We showed that certain classes of hypergraph subsytem codes have poor distance. 
These results hint that some configurations to be avoided for good hypergraph subsystem codes. 
We also gave a new construction for subsystem surface codes. This gave a family of subsystem surface codes with lower overheads.  


\appendix

\section{Uniform rank-3 hypergraph subsystem codes}
\label{app:sq_oct_code}
We give an example of the construction of subsystem codes from Construction~\ref{alg:hg-tsc-construction} in Fig.~\ref{fig:sq_oct_ur3_hsc_const}. 
The 2-colex is obtained in this figure is obtained from Construction~\ref{alg:tcc-construction}.
The assignment of the gauge operators differs from a subsystem code obtained from vertex expansion, see Fig.~\ref{tsc2color}. 

The resulting subsystem code has  four types of stabilizers. 
These are shown in Fig.~\ref{fig:ur3hsc_sq_oct}. 
\begin{compactenum}[(i)]
\item  A cycle consisting of simple edges surrounding an $f$-face of $\Gamma_2$.
\item  A cycle formed by alternating simple and hyperedges on the outer boundary of a $v$-face and alternating simple edges on the boundary of its inner face introduced by construction~\ref{alg:hg-tsc-construction}.
\item A cycle  consisting of alternating simple edges around an $f$-face and the simple and hyperedges belonging to the surrounding $v$- and $f$- faces. 
\item  A loop of simple edges forming the inner face of a $v$-face.
\end{compactenum}
As in Algorithm~\ref{alg:tcc-projection1}, the bit flip errors are corrected first using the stabilizers from rank-2 cycles.
This leaves $Z$ errors on the hyperedges. 
To correct the phase flip errors, the subsystem code is then mapped to the color code based on square octagon lattice as shown in Fig.~\ref{fig:tsc2color_sqoct} underlying the subsystem code. 
The error estimated after decoding the color code is then 
lifted to the subsystem code. For every $Z$ error estimated in the color code the corresponding error estimate on the subsystem code will be the $Z$ errors on the three qubits of the corresponding hyperedge.
\begin{figure*}[h!]
\centering
  \includegraphics{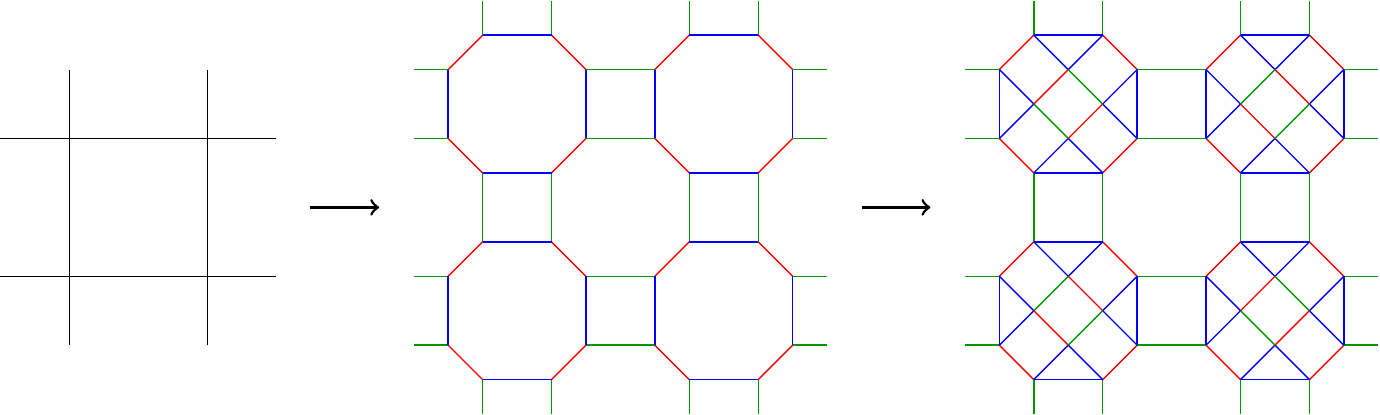}
\caption{Construction of uniform rank-3 hypergraph subsystem codes based on square lattice: A 2-colex is constructed from the square lattice using Construction~\ref{alg:tcc-construction}.
The subsystem code is constructed from the 2-colex using Construction~\ref{alg:hg-tsc-construction}}
\label{fig:sq_oct_ur3_hsc_const}
\end{figure*}

\begin{figure*}[h!]
\centering
\begin{subfigure}[t]{0.3\textwidth}
\centering
  \includegraphics{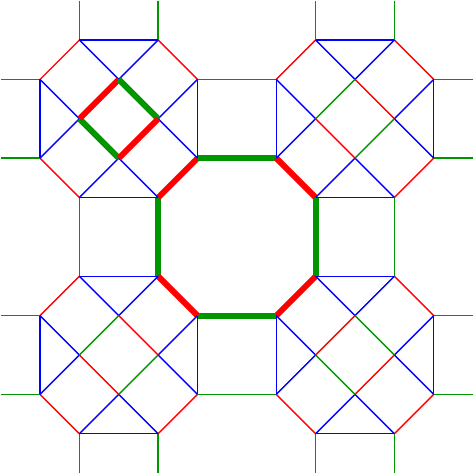}
\caption{The cycles formed by thick edges show the simple cycles around respective faces}
\end{subfigure}%
~
\begin{subfigure}[t]{0.3\textwidth}
\centering
  \includegraphics{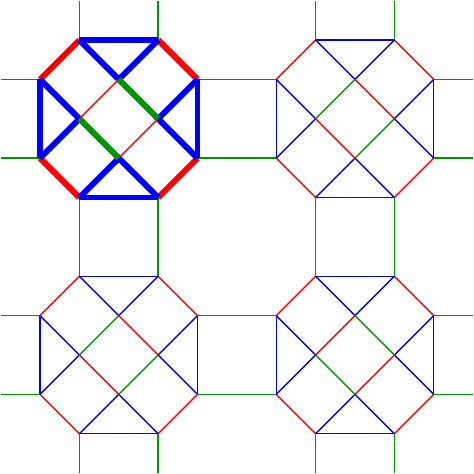}
\caption{The thick simple edges and thick hyper edges constitute a hyper cycle around the square face}
\end{subfigure}
~
\begin{subfigure}[t]{0.3\textwidth}
\centering
  \includegraphics{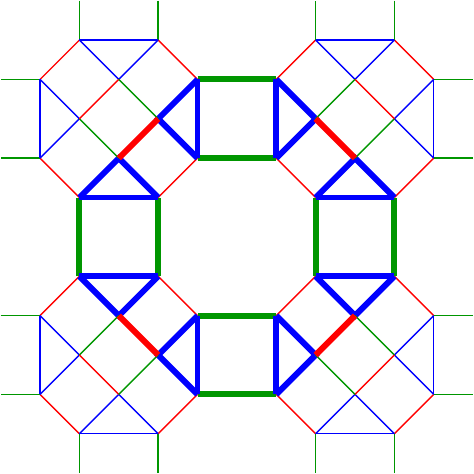}
\caption{The hyperedges highlighted with thick blue edges and thick simple edges form a hypercycle around the octagon face}
\end{subfigure}

\caption{(Color online) Illustrating two types of cycles associated with the stabilizers of uniform rank-3 hypergraph subsystem codes constructed from square lattice.}
\label{fig:ur3hsc_sq_oct}
\end{figure*}

\begin{figure}[h]
\centering
  \includegraphics{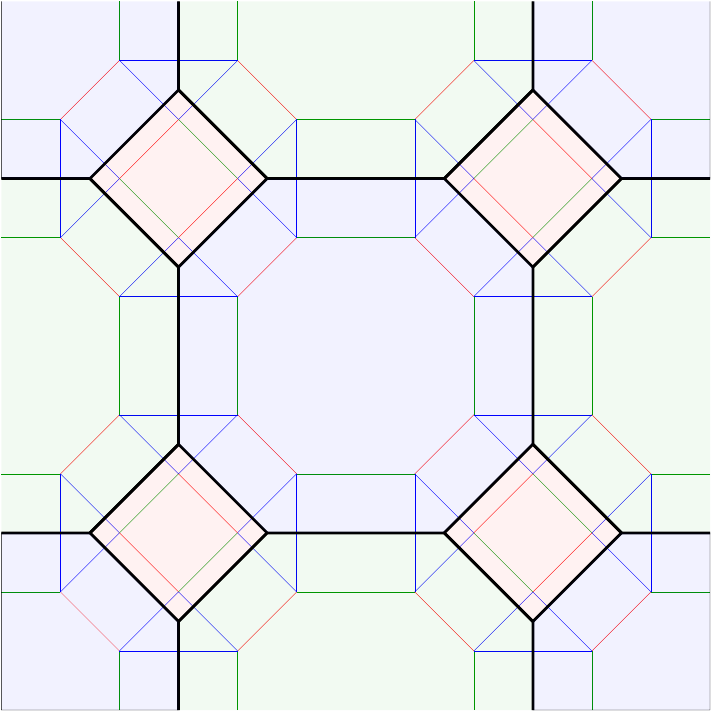}
\caption{(Color online) Mapping the uniform rank-3 hypergraph  constructed in Fig.~\ref{fig:sq_oct_ur3_hsc_const} to a 2-colex, by contracting the rank-3 edges and removing parallel edges.}
\label{fig:tsc2color_sqoct}
\end{figure}

\section*{Acknowledgment}
The authors would like to thank Amit Anil Kulkarni for assistance with the simulations. PS would like to thank David Poulin for helpful discussions on subsystem codes. 

\ifCLASSOPTIONcaptionsoff
  \newpage
\fi

\end{document}